\def\IEEEsubmission{0}

\def\complexNumbers{\mathbb{C}}
\def\realNumbers{\mathbb{R}}

\def\realsNonnegative{\mathbb{R}_{0}^{+}}
\def\realsNonnegative{\mathbb{R}_{0}^{+}}
\def\integers{\mathbb{Z}}
\def\integersPositiveSet{\mathbb{Z}^{+}}
\def\integersNonnegativeSet{\mathbb{Z}_{0}^{+}}

\def\constante{{\rm e}}
\def\constantj{{\rm j}}

\def\speedoflight{c}
\def\lengthGaGb{M}
\def\CPSize{N_{\rm CP}}
\def\fsample{f_{\rm sample}}
\def\Tsample{T_{\rm sample}}
\def\noisePower{N_0}

\def\numberoftargets{R}
\def\symbolDuration{T_{\rm s}}
\def\CPDuration{T_{\rm CP}}
\def\symbolEnergy{E_{\rm s}}

\def\delayVar{\tau}
\def\fcarrier{f_{\rm c}}
\def\separationValueMax{\Delta_{\text{no-loss}}}
\def\numberofTxBits{p}
\def\numofSelectorBits{p_1}
\def\numofModulationBits{p_2}
\def\bandwidthchirp{B}
\def\fcarrier{f_{\rm c}}
\def\numberofIndices{L}
\def\spectralEfficiency{\rho}
\def\EbNO{E_{\text{b}}/\noisePower}
\def\distanceminpsk{d_{\rm{psk}}}
\def\distanceminindex{d_{\rm{ind}}}
\def\probabiltyError{P_{\rm e}}
\def\probabiltyErrorPSK{P_{\numberOfPointsForPSK\text{-PSK}}}
\def\separationValue{\Delta}
\def\numberOfOccupiedSubcarriers{D}

\def\numberOfShifts{M}
\def\lowerFrequency{L_{\rm d}}
\def\upperFrequency{L_{\rm u}}
\def\idftSize{N}
\def\OCB{{M}_{\rm ocb}}

\def\numberOfPointsForPSK{H}
\def\averagePower{P_{\rm av}}

\def\SNRpost{{{SNR}_{\rm post}}}

\def\signalFactor{{{\alpha}_{\rm MMSE}}}
\def\distanceSumVar{Z}

\def\positiveInt{m}
\def\naturalnum{k}

\def\pskSerialSequence{\textit{\textbf{h}}}
\def\indexSequence{\textit{\textbf{i}}}
\def\pskSequence{\textit{\textbf{s}}}
\def\pskSerialSequenceDetect{\hat{\textit{\textbf{h}}}}
\def\indexSequenceDetect{\hat{\textit{\textbf{{i}}}}}

\def\chirpSet{\mathbb{W}}

\def\identityMatrix{{\rm \bf I}}
\def\completeMatrix{{\rm \bf W}}
\def\noiseVector{{\rm \bf n}}
\def\fdssVector{{\rm \bf f}}
\def\channelFVector{{\rm \bf h}}
\def\rxSymbolsVector{{\rm \bf b}}
\def\dataVector{{\rm \bf d}}
\def\DFTmtx[#1]{{{\rm \bf D}_{#1}}}
\def\diagonalMatrixFromVector[#1]{\text{diag}\{ {#1} \}}
\def\delayMtx{{{\rm \bf T}}}
\def\delayMtxSlack{{{\rm \bf \dot{T}}}}
\def\amplitudeVectorSlack{{{\rm \bf \dot{a}}}}
\def\amplitudeVector{{\rm \bf a}}
\def\symbolsVectorinFrequency{{\rm \bf w}}
\def\channelFVectorEst{{{\rm \bf \tilde{h}}}}
\def\identityMatrix{{\rm \bf I}}
\def\seqGa{\textit{\textbf{a}}}
\def\seqGb{\textit{\textbf{b}}}
\def\seqGaP{A}
\def\seqGbP{B}

\def\timeVar{t}
\def\freqVar{f}

\def\polyVariable{z}

\def\indexSubcarrier{k}
\def\indexDataSymbolTX{m}
\def\indexDataSymbolRX{l}
\def\indexSample{n}
\def\lagForCorrelation{l}
\def\indexEleOfSeq{i}
\def\indexPSK{z}
\def\indexChirp{\ell}
\def\indexSeparation{q}
\def\targetindex{s}
\def\chirpm{p}
\def\chirpn{r}
\def\indexIteration{n}
\def\indexRecursion{r}

\def\OFDMinTime[#1][#2]{s_{#1}(#2)}

\def\separationDis[#1]{S_{#1}}
\def\distanceIndex[#1]{{\mathcal D}(#1)}
\def\cardinality[#1][#2]{\mathcal{A}_{#1}(#2)}
\def\distanceIndex[#1]{{\mathcal D}(#1)}

\def\datasymbolestimate[#1][#2]{\hat{d}_{{#1},{#2}}}

\def\eventofeeror[#1][#2][#3][#4]{{\rm E}_{{#1},{#2}|{#3},{#4}}}
\def\CardinalityDis[#1][#2]{\mathcal{B}_{#1}(#2)}
\def\probabiltyunionbound[#1]{U_{#1}}
\def\qfunc[#1]{Q\left(#1\right)}
\def\probabFunction[#1]{P\left(#1\right)}
\def\symbolMultiDim[#1][#2]{{d}_{{#1},{#2}}}

\def\timeArrival[#1]{\tau_{#1}}
\def\channelimpulseresponse[#1]{h\left(#1\right)}
\def\channelimpulseresponseTaps[#1]{h_{#1}}
\def\channelfreqresponse[#1]{\zeta_{#1}}
\def\channelfreqresponseC[#1]{\lambda_{#1}}
\def\distance[#1]{r_{#1}}
\def\distanceEst[#1]{\tilde{r}_{#1}}
\def\reflectioncoefficient[#1]{\alpha_{#1}}
\def\reflectioncoefficientEst[#1]{\tilde{\alpha}_{#1}}
\def\reflectioncoefficientSlack[#1]{{\dot{\alpha}}_{#1}}

\def\diracfunction[#1]{\delta\left(#1\right)}
\def\fresnelC[#1]{C(#1)}
\def\fresnelS[#1]{S(#1)}
\def\linearXone{\alpha_\indexSubcarrier}
\def\linearXtwo{\beta_\indexSubcarrier}
\def\linearCoef{\gamma_\indexSubcarrier}

\def\selectedChirpIndex[#1]{i_{#1}}
\def\selectedChirpIndexDetect[#1]{\hat{i}_{#1}}
\def\selectedChirpIndexDomain[#1]{\dot{i}_{#1}}
\def\symbolPSK[#1]{s_{#1}}
\def\symbolPSKdetect[#1]{\hat{s}_{#1}}
\def\symbolPSKdomain[#1]{\dot{s}_{#1}}
\def\symbolPSKserial[#1]{h_{#1}}
\def\symbolPSKserialdetect[#1]{\hat{h}_{#1}}
\def\symbolPSKserialdomain[#1]{\dot{h}_{#1}}
\def\psksymbol[#1]{d_{#1}}
\def\dataSymbolAfterIDFTspread[#1]{\tilde{d}_{#1}}

\def\chirpphase[#1][#2]{{\psi_{#1}{(#2)}}}
\def\seqx[#1]{\textit{{x}}(#1)}
\def\seqy[#1]{\textit{{y}}(#1)}

\def\transmittedSignalDiscrete[#1]{x_{#1}}
\def\transmittedSignal[#1]{x\left(#1\right)}
\def\receivedSignalDiscrete[#1]{r\left[#1\right]}
\def\receivedSignal[#1]{r\left(#1\right)}
\def\receivedSignalDiscreteFrequency[#1]{b_{#1}}
\def\basisFunction[#1]{B_{#1}(\timeVar)}
\def\amountOfShift[#1]{\tau_{#1}}
\def\dataSymbols[#1]{d_{#1}}
\def\angleSignal[#1]{\psi_{#1}(\timeVar)}
\def\angleSignalTime[#1][#2]{\psi_{#1}({#2})}
\def\instantaneousFrequency[#1]{F_{#1}(t)}
\def\besselFunctionFirstKind[#1][#2]{J_{#1}\left(#2\right)}
\def\fourierSeries[#1]{c_{#1}}
\def\FDSScoef[#1]{f_{#1}}
\def\fourier[#1]{\mathcal{F}\{#1\}}
\def\noiseDiscrete[#1]{\eta_{#1}}
\def\noiseVariance{\sigma^2_{\rm n}}

\def\delayVector[#1]{{{\rm \bf t}_{#1}}}
\def\delayElement[#1]{T_{#1}}
\def\timeArrivalEst[#1]{\tilde{\tau}_{#1}}
\def\timeArrivalSlack[#1]{{\dot{\tau}}_{#1}}
\def\symbolsInFrequency[#1]{w_{#1}}

\def\eleGa[#1]{{a}_{#1}}
\def\eleGb[#1]{{b}_{#1}}
\def\apac[#1][#2]{\rho_{#1}(#2)}
\def\apacPositive[#1][#2]{\rho^{+}_{#1}(#2)}
\def\polySeq[#1][#2]{{#1}(#2)}

\def\anumber{x}
\def\anumberAfterSum[#1]{{n(#1)}}
\def\anumberAfterSumB[#1]{{k(#1)}}
\def\anaturalNumberToSeqeunceA{n}
\def\anaturalNumberToSeqeunceB{k}
\def\functionB[#1]{{f(#1)}}
\def\algorithmEncoderB[#1]{{\epsilon_{\mathcal{B}}(#1)}}
\def\algorithmEncoderA[#1]{{\epsilon_{\mathcal{A}}(#1)}}
\def\algorithmDecoderB[#1]{{\epsilon^{-1}_{\mathcal{B}}(#1)}}
\def\algorithmDecoderA[#1]{{\epsilon^{-1}_{\mathcal{A}}(#1)}}

\newcommand\mydots{\hbox to 1em{.\hss.\hss.}}

\if\IEEEsubmission1
\documentclass[journal,12pt,onecolumn,draftclsnofoot,]{IEEEtran}
\else
\documentclass[journal]{IEEEtran}
\fi
\IEEEoverridecommandlockouts
\usepackage{multicol}
\usepackage{lipsum}
\usepackage{mathtools}
\usepackage{amsthm}
\usepackage{acronym}
\usepackage[bookmarksopen=true]{hyperref}
\usepackage{stfloats}
\usepackage{amsfonts}
\usepackage{acronym}
\usepackage{cite}
\usepackage{multirow}
\usepackage{bm}
\usepackage{amsmath,amssymb}
\usepackage{graphicx}
\usepackage[table,xcdraw]{xcolor}
\usepackage[caption=false,font=footnotesize]{subfig}
\usepackage[geometry]{ifsym}
\usepackage{array}
\usepackage[utf8]{inputenc}
\usepackage[T1]{fontenc}
\usepackage{algpseudocode}

\usepackage{hhline}
\usepackage[ruled,vlined]{algorithm2e}

\usepackage{tikz}
\usepackage{lipsum}
\usetikzlibrary{calc,shadings,patterns}
\makeatletter
\tikzset{%
  remember picture with id/.style={%
    remember picture,
    overlay,
    save picture id=#1,
  },
  save picture id/.code={%
    \edef\pgf@temp{#1}%
    \immediate\write\pgfutil@auxout{%
      \noexpand\savepointas{\pgf@temp}{\pgfpictureid}}%
  },
  if picture id/.code args={#1#2#3}{%
    \@ifundefined{save@pt@#1}{%
      \pgfkeysalso{#3}%
    }{
      \pgfkeysalso{#2}%
    }
  }
}

\def\savepointas#1#2{%
  \expandafter\gdef\csname save@pt@#1\endcsname{#2}%
}

\def\tmk@labeldef#1,#2\@nil{%
  \def\tmk@label{#1}%
  \def\tmk@def{#2}%
}

\tikzdeclarecoordinatesystem{pic}{%
  \pgfutil@in@,{#1}%
  \ifpgfutil@in@%
    \tmk@labeldef#1\@nil
  \else
    \tmk@labeldef#1,(0pt,0pt)\@nil
  \fi
  \@ifundefined{save@pt@\tmk@label}{%
    \tikz@scan@one@point\pgfutil@firstofone\tmk@def
  }{%
  \pgfsys@getposition{\csname save@pt@\tmk@label\endcsname}\save@orig@pic%
  \pgfsys@getposition{\pgfpictureid}\save@this@pic%
  \pgf@process{\pgfpointorigin\save@this@pic}%
  \pgf@xa=\pgf@x
  \pgf@ya=\pgf@y
  \pgf@process{\pgfpointorigin\save@orig@pic}%
  \advance\pgf@x by -\pgf@xa
  \advance\pgf@y by -\pgf@ya
  }%
}

\makeatother

\newcounter{hatchNumber}
\DeclarePairedDelimiter\floor{\lfloor}{\rfloor}

\usepackage{cuted}
\makeatletter
\newif\ifAC@uppercase@first%
\def\Aclp#1{\AC@uppercase@firsttrue\aclp{#1}\AC@uppercase@firstfalse}%
\def\AC@aclp#1{%
	\ifcsname fn@#1@PL\endcsname%
	\ifAC@uppercase@first%
	\expandafter\expandafter\expandafter\MakeUppercase\csname fn@#1@PL\endcsname%
	\else%
	\csname fn@#1@PL\endcsname%
	\fi%
	\else%
	\AC@acl{#1}s%
	\fi%
}%
\def\Acp#1{\AC@uppercase@firsttrue\acp{#1}\AC@uppercase@firstfalse}%
\def\AC@acp#1{%
	\ifcsname fn@#1@PL\endcsname%
	\ifAC@uppercase@first%
	\expandafter\expandafter\expandafter\MakeUppercase\csname fn@#1@PL\endcsname%
	\else%
	\csname fn@#1@PL\endcsname%
	\fi%
	\else%
	\AC@ac{#1}s%
	\fi%
}%
\def\Acfp#1{\AC@uppercase@firsttrue\acfp{#1}\AC@uppercase@firstfalse}%
\def\AC@acfp#1{%
	\ifcsname fn@#1@PL\endcsname%
	\ifAC@uppercase@first%
	\expandafter\expandafter\expandafter\MakeUppercase\csname fn@#1@PL\endcsname%
	\else%
	\csname fn@#1@PL\endcsname%
	\fi%
	\else%
	\AC@acf{#1}s%
	\fi%
}%
\def\Acsp#1{\AC@uppercase@firsttrue\acsp{#1}\AC@uppercase@firstfalse}%
\def\AC@acsp#1{%
	\ifcsname fn@#1@PL\endcsname%
	\ifAC@uppercase@first%
	\expandafter\expandafter\expandafter\MakeUppercase\csname fn@#1@PL\endcsname%
	\else%
	\csname fn@#1@PL\endcsname%
	\fi%
	\else%
	\AC@acs{#1}s%
	\fi%
}%
\edef\AC@uppercase@write{\string\ifAC@uppercase@first\string\expandafter\string\MakeUppercase\string\fi\space}%
\def\AC@acrodef#1[#2]#3{%
	\@bsphack%
	\protected@write\@auxout{}{%
		\string\newacro{#1}[#2]{\AC@uppercase@write #3}%
	}\@esphack%
}%
\def\Acl#1{\AC@uppercase@firsttrue\acl{#1}\AC@uppercase@firstfalse}
\def\Acf#1{\AC@uppercase@firsttrue\acf{#1}\AC@uppercase@firstfalse}
\def\Ac#1{\AC@uppercase@firsttrue\ac{#1}\AC@uppercase@firstfalse}
\def\Acs#1{\AC@uppercase@firsttrue\acs{#1}\AC@uppercase@firstfalse}

\newtheorem{theorem}{Theorem}
\newtheorem{definition}{Definition}
\newtheorem{lemma}{Lemma}

\newtheorem{corollary}{Corollary}
 
\newtheorem{example}{\color{black} Example} 

\acrodef{SC}{single-carrier}
\acrodef{UWB}{ultra-wide band }
\acrodef{DQPSK}{Differential Quadrature Phase Shift Keying}
\acrodef{OOB}{out-of-band}
\acrodef{PMEPR}{peak-to-mean envelope power ratio}
\acrodef{SIC}{successive interference cancellation}
\acrodef{PAPR}{peak-to-average-power ratio}
\acrodef{APAC}[AACF]{aperiodic autocorrelation function}
\acrodef{OFDM}{orthogonal frequency division multiplexing}
\acrodef{DFT}{discrete Fourier transform}
\acrodef{DC}{direct current}
\acrodef{CS}{complementary sequence}
\acrodef{GCP}{Golay complementary pair}
\acrodef{ANF}{algebraic normal form}
\acrodef{PSK}{phase-shift keying}
\acrodef{QAM}{quadrature amplitude modulation}
\acrodef{QPSK}{quadrature \ac{PSK}}
\acrodef{GDJ}{Golay-Davis-Jedwab}
\acrodef{FFT}{fast Fourier transform}
\acrodef{BER}{bit error rate}
\acrodef{SNR}{signal-to-noise ratio}
\acrodef{4G}{Fourth Generation}
\acrodef{5G}{Fifth Generation}
\acrodef{NR}{New Radio}
\acrodef{LTE}{Long-Term Evolution}
\acrodef{PTS}{partial transmit sequences}
\acrodef{PSD}{power spectral density}
\acrodef{LDPC}{low-density parity-check}
\acrodef{SE}{spectral efficiency}
\acrodef{eLAA}{enhanced licensed-assisted access}
\acrodef{NR-U}{NR-Unlicensed}
\acrodef{RM}{Reed-Muller}
\acrodef{AE}{autoencoder}
\acrodef{DNN}{deep neural network}
\acrodef{OFDM-AE}{OFDM-based autoencoder}
\acrodef{DL}{deep learning}
\acrodef{CP}{cyclic prefix}
\acrodef{AWGN}{additive white Gaussian noise}
\acrodef{P2C}{polar-to-Cartesian}
\acrodef{CFR}{channel frequency response}
\acrodef{CIR}{channel impulse response}
\acrodef{ReLU}{rectified linear unit}
\acrodef{LMMSE}{linear minimum mean square error}
\acrodef{BPSK}{binary phase-shift keying}
\acrodef{BLER}{block error rate}
\acrodef{ML}{maximum-likelihood}
\acrodef{PHY}{physical layer}
\acrodef{PA}{power amplifier}
\acrodef{IDFT}{inverse \ac{DFT}}
\acrodef{DoF}{degrees-of-freedom}
\acrodef{IoT}{Internet-of-Things}
\acrodef{DFT-s-OFDM}{discrete Fourier transform spread orthogonal frequency division multiplexing}
\acrodef{MMSE}{minimum mean square error}
\acrodef{FDE}{frequency-domain equalization}
\acrodef{FrFT}{fractional Fourier transform}
\acrodef{TF}{time-frequency}
\acrodef{BFSK}{binary frequency-shift keying}
\acrodef{CSS}{chirp-spread spectrum}
\acrodef{BCSS}{binary chirp spread spectrum}
\acrodef{EVA}{Extended Vehicular A}
\acrodef{MIMO}{multi-input multi-output}
\acrodef{PIC}{parallel interference cancellation}
\acrodef{LoRa}{Long Range}
\acrodef{HF}{high-frequency}
\acrodef{FDSS}{frequency-domain spectral shaping}
\acrodef{OCB}{occupied channel bandwidth}
\acrodef{FSK}{frequency-shift keying}
\acrodef{RF}{radio-frequency}
\acrodef{IM}{index modulation}
\acrodef{DFRC}{dual-function radar and communication}
\acrodef{ISI}{inter-symbol interfence}
\acrodef{iid}[i.i.d.]{independent and identically distributed}
\acrodef{CSC-IM}[CSC-IM]{\ac{IM} with \acp{CSC}}
\acrodef{DFT-s-OFDM-IM}{IM with \ac{DFT-s-OFDM}}
\acrodef{OFDM-IM}{IM with OFDM}
\acrodef{BS}{base station}
\acrodef{MF}{matched filter}
\acrodef{CSC}[CSC]{circularly-shifted chirp}
\acrodef{FMCW}{frequency-modulated continuous-wave}
\acrodef{IS}{index separation}
\acrodef{RMSE}{root-mean-square error}
\acrodef{CE}{constant-envelope}
\acrodef{TDRW}{time-diversity radar waveform}
\acrodef{RX}{receiver}
\acrodef{RXr}{radar receiver}
\acrodef{RXc}{communication receiver}
\acrodef{TX}{transmitter}
\acrodef{AoD}{angle-of-departure}
\acrodef{AoA}{angle-of-arrival}
\acrodef{CRLB}{Cramer-Rao lower bound}
\acrodef{FIM}{Fisher information matrix}
\acrodef{OTFS}{orthogonal-time-frequency-space modulation}
\acrodef{UB}{union bound}
\acrodef{ADC}{analog-to-digital converter}
\acrodef{AC}{autocorrelation}
\acrodef{CS-RM}[CSs-RM]{\acp{CS} based on \ac{RM} code}

\def\BibTeX{{\rm B\kern-.05em{\sc i\kern-.025em b}\kern-.08em
    T\kern-.1667em\lower.7ex\hbox{E}\kern-.125emX}}
\begin{document}

\title{
{Index Modulation with Circularly-Shifted Chirps for Dual-Function Radar and Communications
}\\
\thanks{Alphan~\c{S}ahin and Safi Shams Muhtasimul Hoque are with the Electrical  Engineering Department,
	University of South Carolina, Columbia, SC, USA. Chao-Yu Chen is with Department of Engineering Science, National Cheng Kung University, Tainan, Taiwan, R.O.C.  E-mails: asahin@mailbox.sc.edu, shoque@email.sc.edu, super@mail.ncku.edu.tw}
\thanks{This paper was presented in part at the IEEE Global Communications Conference - Workshop on Integrated Sensing and Communications 2020  \cite{Safi_2020_GC} and	IEEE Consumer Communications \& Networking Conference 2021 \cite{Safi_2020_CCNC}.}
\author{Alphan~\c{S}ahin, Safi Shams Muhtasimul Hoque, and Chao-Yu Chen} 
}

\maketitle

\begin{abstract}
	In this study, we propose \ac{CSC-IM} for \ac{DFRC} systems.  The proposed scheme encodes the information bits with the  \ac{CSC} indices and the \ac{PSK} symbols. It allows the receiver to exploit the frequency selectivity naturally in fading channels by combining \ac{IM} and wideband \acp{CSC}. It also leverages the fact that a \ac{CSC} is a constant-envelope signal to achieve a controllable \ac{PMEPR}. For radar functionality,  \ac{CSC-IM}  maintains the good \ac{AC} properties of a chirp by ensuring that  the transmitted \acp{CSC} are separated apart sufficiently in the time domain through \ac{IS}. We investigate the impact of \ac{IS} on \ac{SE} and obtain the corresponding mapping functions. For theoretical results,  we derive the \ac{UB} of the \ac{BLER} for arbitrary chirps and the \acp{CRLB} for the range and reflection coefficients for the \ac{MF}-based estimation. We also prove that \acp{CS} can be constructed through \acp{CSC} by linearly combining the Fourier series of \acp{CSC}. Finally, through comprehensive comparisons, we demonstrate the efficacy of the proposed scheme for \ac{DFRC} scenarios.
\end{abstract}

\acresetall
\begin{IEEEkeywords}
	Chirps, \acp{CS}, \ac{IM}, \ac{PMEPR}
\end{IEEEkeywords}

\acresetall

\section{Introduction}
The merging of communications and radar functionalities into a single wireless network can improve the efficient utilization of the physical resources and address the potential interference issues between radar and communication systems \cite{Paul_2017}. 
To realize such a network, the transmitted signals need to be designed based on two different objectives. For radar functionality, the primary goal is to improve the accuracy of the range and/or velocity estimations. On the other hand, the waveform for communications is optimized by considering communications-related metrics such as error rate or data rate. Hence, developing a transmission scheme suitable for both features is not trivial \cite{Ma_2020spm}. In this study, we address this issue and propose  \ac{CSC-IM} for \ac{DFRC} systems, which can be synthesized with \ac{DFT-s-OFDM} used in 3GPP \ac{5G} \ac{NR} \cite{nr_general_2021} and \ac{4G} \ac{LTE} uplink \cite{lte_general_2018}.

Chirps are prominent for radar applications due to their excellent \ac{PMEPR} and good \ac{AC} properties in time-varying channels. They also facilitate the radar implementation with a simple hardware architecture through correlations in \ac{RF}. They were first proposed in \cite{darlingtonPatent1949} to achieve a long-range high-resolution radar, which was later extended to encode information through the slope of the chirps in the \ac{TF} plane.
For developing more sophisticated multiplexing methods based on chirps,  the bases constructed through chirps have been studied  in the literature, extensively. For example, in \cite{nozh2}, an orthogonal amplitude-variant linear chirp set where each chirp has a different chirp rate was proposed.  In \cite{fresnel_nozh_7523229}, orthogonal chirps were constructed by shifting the chirps in the frequency domain. In \cite{sahin_2020}, \acp{CSC} were proposed by using the structure of \ac{DFT-s-OFDM} with a specific \ac{FDSS} function. It was shown that \acp{CSC} can be transmitted simultaneously in an overlapping manner as opposed to linear chirps.  However, these studies do not particularly consider \ac{DFRC} applications.

Another way of constructing a scheme suitable for \ac{DFRC} is to utilize the multiplexing methods that are highly used in wireless communication systems. In \cite{Paul_2017}, \ac{OFDM}, chirps, and \acp{CS} with \ac{SC} waveform and \ac{OFDM}  \cite{davis_1999,sahin_2020gm} are surveyed for \ac{DFRC} scenarios. It is emphasized that the noise-like nature of \ac{OFDM} signals in time can be beneficial for a typical radar at the expense of  high \ac{PMEPR} and the sidelobe growth in the \ac{AC}  function due to the existence of \ac{CP}. On the other hand, \ac{OFDM} provides an excellent framework for the \ac{CFR}-based estimation methods as the processing directly occurs in the modulation domain. For example, in \cite{Sturm_2009}, several range profiles were demonstrated through \ac{CFR} by using \ac{OFDM} as a radar waveform. In \cite{sturm_2011}, \ac{OFDM} is compared with chirps and other spread-spectrum techniques in detail. In \cite{Braun_2010} and \cite{Bica_2017_radarconf}, \ac{ML}-based range and velocity estimators for a single target are investigated for \ac{OFDM}. 
An iterative algorithm based on filtering and clipping was investigated in \cite{Sharma_2019} to reduce the \ac{PMEPR}  of the multicarrier radar waveform at a cost of the distorted \ac{AC} function.
In \cite{Turlapaty_2014}, the  weighted subcarriers were proposed to mitigate the \ac{PMEPR} for an \ac{OFDM}-based radar. The authors in \cite{Sahan2020} proposed to transmit optimized complex-valued data through the unused subcarriers of \ac{OFDM} for radar functionality.
In \cite{Sen_2011}, a real-valued baseband OFDM signal that  modulates the phase of the carrier, called constant-envelope \ac{OFDM}, was proposed for detecting a target in a multi-path scenario. In \cite{Kumari_2018}, the  IEEE 802.11ad \ac{SC} preamble based on \acp{CS} was exploited for \ac{DFRC} applications.

Recently, \ac{IM} receives attention for \ac{DFRC} scenarios as it promises communications with minimal degradation to the radar performance \cite{Ma_2020spm}. \Ac{IM} is a permutation modulation \cite{Slepian_1965} and encodes the information in the order of discrete objects, e.g., antennas, subcarriers, or time slots. For a comprehensive survey on \ac{IM} and its applications, we refer the reader to the surveys in \cite{Cheng_2018,Ishikawa_2018,Miaowen_2019,Tianqi_2019} and the references therein. For communications applications, in \cite{basar_2013}, \ac{OFDM-IM} is investigated by grouping the subcarriers and activating or deactivating subcarriers within the groups for encoding information.
In \cite{Mesleh_2008}, extra information bits were transmitted by turning on and off the antennas, i.e., spatial modulation. In  \cite{Khandani_2013}, media-based modulation was proposed by setting the on/off status of available RF mirrors. For radar applications, in \cite{BouDaher2016}, the information was encoded by shuffling the radar signals for a \ac{MIMO} radar setup. In \cite{Ma2020}, it was proposed to select a subset of subcarriers and/or antennas. In this study, with the same motivation of minimum degradation to radar, we utilize \ac{IM} with \acp{CSC}.

\subsection{Contributions}
In this study, we provide both theoretical and practical contributions listed as follows:
\begin{itemize}
	\item {\bf{A scheme for \ac{DFRC} systems}:}
	We propose to transmit multiple modulated \acp{CSC} simultaneously for a \ac{DFRC} system, where the information bits are encoded with indices and \ac{PSK} symbols. Since it relies on the structure of \ac{DFT-s-OFDM} with  a special \ac{FDSS} and \ac{IM}, it leads to low-complexity transmitters and receivers. 
	The main advantage of the proposed scheme is controllable low-\ac{PMEPR} and \ac{SE} while still being a wideband signal needed for radar functionality at \ac{RXr} and exploiting frequency selectivity at \ac{RXc}. 
	\item {\bf Theoretical bounds and relationships:} We establish a connection between \acp{CS} and chirps.	We derive the \ac{UB} of \ac{BLER} for the proposed scheme. Also, for radar functionality, we  obtain the \acp{CRLB} for ranges and reflection coefficients, which consider the phase of the reflected signal as a function of the target's range.
	\item {\bf Improved range estimation with Index Separation:}
	We investigate two  range estimation methods: \Ac{MF}-based and \ac{LMMSE}-based estimations. To increase the estimation accuracy, we introduce a concept called {\em \ac{IS}} that ensures a low \ac{AC}  zone. We derive the maximum separation between \acp{CSC} in time without affecting the \ac{SE} for any number of chirps, theoretically. To facilitate the encoder and decoder with the \ac{IS}, we also develop methods that construct a bijective mapping between the information bits and indices.
	\item {\bf Comprehensive comparisons:}
	The proposed scheme is compared with \ac{OFDM-IM}, \ac{DFT-s-OFDM-IM}, and the \ac{CS-RM}, comprehensively, in terms of \ac{PMEPR}, error rate, estimation accuracy, and radar resolution in various scenarios, which provides further insights into \ac{DFRC} waveform.
\end{itemize}
The rest of the paper is organized as follows. In Section~\ref{sec:prelim},  the system model is provided. In Section~\ref{sec:scheme}, we introduce \ac{CSC-IM} and derive the \ac{UB} of \ac{BLER}. The relationship between \acp{CS} and chirps and the trade-off between \ac{SE} and \ac{PMEPR} are  discussed in this section. In \ref{sec:radar}, we analyze the radar functionality with \ac{CSC-IM}. We discuss range estimation and \ac{IS}. In Section \ref{sec:numresults},  numerical results are presented. The paper is concluded in Section \ref{sec:conclusion}.


\section{Preliminaries and System Model}
\label{sec:prelim}

The sets of complex numbers, real numbers, non-negative real numbers, positive integers, non-negative integers, and the set of integers $\{0,1,\mydots,\numberOfPointsForPSK-1\}$ are denoted by $\complexNumbers$, $\realNumbers$, $\realsNonnegative$, $\integersPositiveSet$, $\integersNonnegativeSet$, and $\integers_\numberOfPointsForPSK$ respectively. Conjugation is denoted by $(\cdot)^*$. The notation $ (\eleGa[0],\eleGa[1],\dots, \eleGa[\lengthGaGb-1])$ represents the sequence $\seqGa$. The constant $\constantj$ denotes $\sqrt{-1}$.

\subsection{Circularly-Shifted Chirps}
\label{subsec:csc}
Let $\angleSignal[0]\in\realNumbers$ be a periodic function with the period of  $\symbolDuration$. We then define the function $\basisFunction[\indexDataSymbolTX]$ by setting $\basisFunction[\indexDataSymbolTX]=\constante^{\constantj\angleSignal[\indexDataSymbolTX]}$, where $\angleSignal[\indexDataSymbolTX]=\angleSignalTime[0][{\timeVar-\amountOfShift[\indexDataSymbolTX]}]$ and  $\amountOfShift[\indexDataSymbolTX]$ is the amount of translation for $\indexDataSymbolTX\in\{0,1,\mydots, \numberOfShifts-1\}$ and $\amountOfShift[0]=0$. 
The function  ${\basisFunction[\indexDataSymbolTX]}$ for $\timeVar\in[0,\symbolDuration)$ is a \ac{CSC} and it is equal to the circularly-shifted version of the function ${\basisFunction[0]}$ by $\amountOfShift[\indexDataSymbolTX]$ seconds.

The Fourier series expansion of $\basisFunction[\indexDataSymbolTX]$ can be obtained as
\begin{align}
\basisFunction[\indexDataSymbolTX]  \approx \sum_{\indexSubcarrier=\lowerFrequency}^{\upperFrequency} \fourierSeries[\indexSubcarrier] \constante^{\constantj2\pi\indexSubcarrier\frac{\timeVar-\amountOfShift[\indexDataSymbolTX]}{\symbolDuration}}~,
\label{eq:basisDecompose}
\end{align}
where $\lowerFrequency<0$ and $\upperFrequency>0$ are integers, and $\fourierSeries[\indexSubcarrier]$ is the $\indexSubcarrier$th Fourier coefficient given by
\begin{align}
\fourierSeries[\indexSubcarrier] = \fourier[{\constante^{\constantj\angleSignal[0]}}] \triangleq \frac{1}{\symbolDuration}\int_{\symbolDuration}
\constante^{\constantj\angleSignal[0]}
\constante^{-\constantj2\pi\indexSubcarrier\frac{\timeVar}{\symbolDuration}}d\timeVar~.
\end{align}

Let ${\numberOfOccupiedSubcarriers}/{2\symbolDuration}$ denote the maximum frequency deviation around the carrier frequency. For an accurate approximation of the right-hand side of \eqref{eq:basisDecompose} to $\basisFunction[\indexDataSymbolTX]$, we assume that $\lowerFrequency<-\numberOfOccupiedSubcarriers/2$ and $\upperFrequency>\numberOfOccupiedSubcarriers/2$. This is due to the fact that $\basisFunction[\indexDataSymbolTX]$ is a function where  $|\fourierSeries[\indexSubcarrier]|$  approaches zero rapidly for $|\indexSubcarrier|>{\numberOfOccupiedSubcarriers}/{2}$. Note that the actual bandwidth of a chirp  is slightly larger than twice the maximum frequency deviation \cite{proakisfundamentals}. It can be calculated based on the total integrated power of the transmitted spectrum, i.e., \ac{OCB}. In this study, we express the \ac{OCB} as $\OCB/\symbolDuration$ Hz, where $\OCB$ is assumed to be a positive integer. Also, the
instantaneous frequency of $\basisFunction[\indexDataSymbolTX]$ around the carrier frequency $\fcarrier$ can  be obtained as $\instantaneousFrequency[\indexDataSymbolTX]=\frac{1}{2\pi}d{\angleSignal[\indexDataSymbolTX]}/{d\timeVar}$ Hz.

In \cite{sahin_2020}, \acp{CSC} are utilized for data transmission by using \ac{DFT-s-OFDM} as follows: Consider a baseband signal given by
$
\transmittedSignal[\timeVar] = \frac{1}{\sqrt{\numberOfShifts}} \sum_{\indexDataSymbolTX=0}^{\numberOfShifts-1} \dataSymbols[\indexDataSymbolTX]\basisFunction[{\indexDataSymbolTX}]
$, 
where  $\dataSymbols[\indexDataSymbolTX]$ denotes the $\indexDataSymbolTX$th modulation symbol. If $\amountOfShift[\indexDataSymbolTX]$ is chosen as $\amountOfShift[\indexDataSymbolTX]=\indexDataSymbolTX/\numberOfShifts\times\symbolDuration$, i.e., uniform spacing in time, by sampling $\transmittedSignal[\timeVar]$ at  $\indexSample/\idftSize\times\symbolDuration$ for $\indexSample\in\{0,1,\mydots\idftSize-1\}$ and using \eqref{eq:basisDecompose}, the baseband signal $\transmittedSignal[\timeVar]$ in discrete time can be written  as
\begin{align}
\transmittedSignalDiscrete[\indexSample]=&\underbrace{\sum_{\indexSubcarrier=\lowerFrequency}^{\upperFrequency}\underbrace{\fourierSeries[\indexSubcarrier]\underbrace{\frac{1}{\sqrt{\numberOfShifts}}\sum_{\indexDataSymbolTX=0}^{\numberOfShifts-1} \dataSymbols[\indexDataSymbolTX] 
			\constante^{-\constantj2\pi \indexSubcarrier \frac{\indexDataSymbolTX}{\numberOfShifts}}}_{ \text{Normalized }\numberOfShifts\text{-point DFT}}}_{\text{Frequency-domain spectral shaping}}
	\constante^{\constantj2\pi \indexSubcarrier \frac{\indexSample}{\idftSize}}}_{\idftSize\text{-point IDFT with zero-padding} }~,
\label{eq:chirpWave}
\end{align}
where $\idftSize>\numberOfShifts=\upperFrequency-\lowerFrequency+1>\numberOfOccupiedSubcarriers$. Hence, by prepending a \ac{CP} with the duration of $\CPDuration$ to the signal in \eqref{eq:chirpWave}, a typical \ac{DFT-s-OFDM} transmitter or receiver can  be utilized for synthesizing modulated \acp{CSC} by only introducing a special \ac{FDSS} filter, i.e., $\{\fourierSeries[\indexSubcarrier]|\indexSubcarrier=\lowerFrequency,\mydots,\upperFrequency\}$. 

In \cite{sahin_2020}, several closed-form expressions for $\fourierSeries[\indexSubcarrier]$ are also provided. For example, let $\instantaneousFrequency[0]$ be a function changing from $-\frac{\numberOfOccupiedSubcarriers}{2\symbolDuration}$ Hz to $\frac{\numberOfOccupiedSubcarriers}{2\symbolDuration}$ Hz, i.e., $\instantaneousFrequency[0]=\frac{\numberOfOccupiedSubcarriers}{2\symbolDuration}\left(\frac{2\timeVar}{\symbolDuration}-1\right)$, i.e., a linear chirp. The $\indexSubcarrier$th Fourier coefficient for the linear chirp can be calculated as 
\begin{align}
\fourierSeries[\indexSubcarrier] = \linearCoef(\fresnelC[{\linearXone}] + \fresnelC[{\linearXtwo}] + \constantj\fresnelS[{\linearXone}] + \constantj\fresnelS[{\linearXtwo}])~,
\label{eq:fresnekfcn}
\end{align}
where $\fresnelC[{\cdot}]$ and  $\fresnelS[{\cdot}]$ are the Fresnel integrals with cosine and sine functions, respectively, and $\linearXone=(\numberOfOccupiedSubcarriers/2+2\pi\indexSubcarrier)/\sqrt{\pi\numberOfOccupiedSubcarriers }$, $\linearXtwo=(\numberOfOccupiedSubcarriers/2-2\pi\indexSubcarrier)/\sqrt{\pi\numberOfOccupiedSubcarriers}$, $\linearCoef= \sqrt{\frac{\pi}{\numberOfOccupiedSubcarriers}}\constante^{-\constantj\frac{(2\pi\indexSubcarrier)^2}{2\numberOfOccupiedSubcarriers}-\constantj\pi \indexSubcarrier}$. For sinusoidal chirps, $\instantaneousFrequency[0]=\frac{\numberOfOccupiedSubcarriers}{2\symbolDuration}\cos\left({2\pi \frac{\timeVar}{\symbolDuration}}\right)$ and it can be shown that
\begin{align}
    \fourierSeries[\indexSubcarrier] = \besselFunctionFirstKind[\indexSubcarrier][\frac{\numberOfOccupiedSubcarriers}{2}],
    \label{eq:besselfcn}
\end{align}
where $\besselFunctionFirstKind[\indexSubcarrier][\cdot]$ is the Bessel function of the first kind of order $\indexSubcarrier$.

\subsection{Complementary Sequences}
\label{subsec:complementarySequence}
A sequence pair  $(\seqGa,\seqGb)$ of length  $\lengthGaGb$ is a \ac{GCP} if the \acp{APAC} of the sequences $\seqGa$ and $\seqGb$ sum to zero for all non-zero lags \cite{Golay_1961}, i.e., $\apac[\seqGa][\lagForCorrelation]+\apac[\seqGb][\lagForCorrelation] = 0$ for $\lagForCorrelation~\neq0$, where $\apac[\seqGa][\lagForCorrelation]$ and $\apac[\seqGb][\lagForCorrelation]$ are the \acp{APAC} of the sequences $\seqGa$ and $\seqGb$ at the $\lagForCorrelation$th lag, respectively. 
 The \ac{APAC} of the sequence $\seqGa$ can be calculated as
 \begin{align}
 	\apac[\seqGa][\lagForCorrelation]\triangleq
 	\begin{cases}
 		\sum_{\indexEleOfSeq=0}^{\lengthGaGb-\lagForCorrelation-1} \eleGa[\indexEleOfSeq]^*\eleGa[\indexEleOfSeq+\lagForCorrelation], & 0\le\lagForCorrelation\le\lengthGaGb-1\\
 		\sum_{\indexEleOfSeq=0}^{\lengthGaGb+\lagForCorrelation-1} \eleGa[\indexEleOfSeq]\eleGa[\indexEleOfSeq-\lagForCorrelation]^*, & -\lengthGaGb+1\le\lagForCorrelation<0\\
 		0,& \text{otherwise}
 	\end{cases}~.
 \end{align}
Each sequence in a \ac{GCP} is called  a \ac{CS}. A GCP $(\seqGa,\seqGb)$ can equivalently be defined as any sequence pair satisfying
${|\polySeq[\seqGaP][{\polyVariable}]|^2+|\polySeq[\seqGbP][{\polyVariable}]|^2} ={\apac[\seqGa][0]+\apac[\seqGb][0]}$, where $\polySeq[\seqGaP][\polyVariable] \triangleq \eleGa[\numberOfShifts-1]\polyVariable^{\numberOfShifts-1} + \eleGa[\numberOfShifts-2]\polyVariable^{\numberOfShifts-2}+ \dots + \eleGa[0]$ and $\polySeq[\seqGbP][\polyVariable] \triangleq \eleGb[\numberOfShifts-1]\polyVariable^{\numberOfShifts-1} + \eleGb[\numberOfShifts-2]\polyVariable^{\numberOfShifts-2}+ \dots + \eleGb[0]$   in indeterminate $\polyVariable$ \cite{parker_2003}. 

Let $\OFDMinTime[\seqGa][\timeVar]=\sum_{\indexEleOfSeq=0}^{\lengthGaGb-1}\eleGa[{\indexEleOfSeq}]\constante^{\constantj2\pi\indexEleOfSeq\frac{\timeVar}{\symbolDuration}}$ for $\timeVar\in[0,\symbolDuration)$ be a  continuous-time baseband \ac{OFDM} symbol generated from a sequence $\seqGa$  with the symbol duration $\symbolDuration$.
It can be shown that the instantaneous peak power of $\OFDMinTime[\seqGa][\timeVar]$ is bounded, i.e., 
$\max_{\timeVar}|\OFDMinTime[\seqGa][\timeVar]|^2 \le \apac[\seqGa][0]+\apac[\seqGb][0]$, if  $\seqGa$ is a \ac{CS}. For this case, the \ac{PMEPR} of the \ac{OFDM} symbol $\OFDMinTime[\seqGa][\timeVar]$, defined as $\max_{\timeVar}|\OFDMinTime[\seqGa][\timeVar]|^2/\averagePower$, is  less than or equal to $10\log_{10}(2)\approx3$~dB if $\averagePower=\apac[\seqGa][0]=\apac[\seqGb][0]$ \cite{paterson_2000}. For non-unimodular \acp{CS}, $\apac[\seqGa][0]$ may not be equal to $\apac[\seqGb][0]$. In that case, the power of an \ac{OFDM} symbol with $\seqGa$ can be different from the one  with $\seqGb$ although the instantaneous peak power is still less than  or equal to $\apac[\seqGa][0]+\apac[\seqGb][0]$. Hence, to avoid misleading results, we define $\averagePower$ as the power of the entire baseband signal in this study.

\subsection{DFRC Scenario}
\begin{figure}[t]
	\centering
	{\includegraphics[width =3.2in]{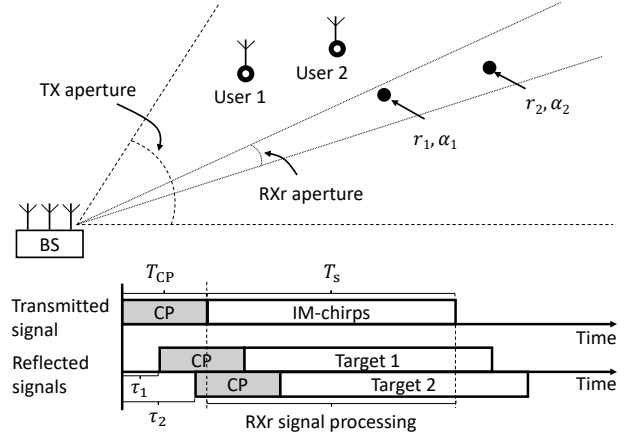}
	} 
	\caption{DFRC scenario and the corresponding timing diagram for the transmitted signal and the radar return for two targets.}
	\label{fig:timingdiagram}
\if\IEEEsubmission1
\vspace{-3mm}
\fi
\end{figure}
In this study, we consider  a \ac{DFRC} scenario where a \ac{BS} broadcasts a communication signal to the users and exploits the same signal to estimate the distances of the surrounding objects as illustrated in \figurename~\ref{fig:timingdiagram}. We assume that a directional antenna at the \ac{RXr} sweeps along azimuth and elevation angles within the aperture of the \ac{TX} to identify the orientation of the targets.

We consider a low-velocity environment, e.g., indoors, where the coherence time is much larger than $\symbolDuration$. Hence, the \ac{CIR} within $\symbolDuration$ can be assumed to be time-invariant and expressed as
$\channelimpulseresponse[\delayVar]=\sum_{\targetindex=1}^{\numberoftargets} \reflectioncoefficient[\targetindex]\diracfunction[{\delayVar-\timeArrival[\targetindex]}]$,
where $\numberoftargets$ is the number of paths, and $\reflectioncoefficient[\targetindex]\in\realNumbers$ and $\timeArrival[\targetindex]\in\realsNonnegative$ are the  gain and the delay of the path \ac{TX}-to-$\targetindex$th target-to-\ac{RXr}, respectively. The path delays can be calculated as $\timeArrival[\targetindex]=2\distance[\targetindex]/\speedoflight$, where $\distance[\targetindex]\in\realNumbers$ is the distance between the $\targetindex$th target and the \ac{BS} and $\speedoflight$ is the speed of light. The \ac{BS}'s goal is to estimate $\distance[\targetindex]$ for $\targetindex\in\{1,\mydots,\numberoftargets\}$ while using the same signal for broadcasting information.

We assume that the \ac{TX} and \ac{RXr} at the \ac{BS} are synchronized in time and there is an ideal phase/frequency synchronization between the \ac{TX} and \ac{RXr} carriers (e.g., fed through the same oscillator). For an \ac{OFDM}-based waveform,  the \ac{CFR} on the $\indexSubcarrier$th subcarrier  can be calculated as
\begin{align} 
	\channelfreqresponse[\indexSubcarrier] &=  \int \channelimpulseresponse[\delayVar]\constante^{-\constantj2\pi\freqVar \delayVar}d\delayVar\Bigr\rvert_{\freqVar=\fcarrier+\frac{\indexSubcarrier}{\symbolDuration}}= \sum_{\targetindex=1}^{\numberoftargets}\reflectioncoefficient[\targetindex]\constante^{-\constantj2\pi\fcarrier\timeArrival[\targetindex]}
	\constante^{-\constantj2\pi\indexSubcarrier\frac{\timeArrival[\targetindex]}{\symbolDuration}},
	\label{eq:radarChannel}
\end{align}
where we assume that  $\timeArrival[\numberoftargets]\le\CPDuration$  and the \ac{RXr} processes the  radar return for $\timeVar\in[\CPDuration,\symbolDuration+\CPDuration)$ so that the path delays include the propagation delays. Therefore, the maximum range of the radar for this specific implementation is  $\speedoflight\times \CPDuration/2$ meters. 
Note that it is possible to increase the maximum range if the \ac{RXr}'s synchronization point is intentionally delayed and chosen based on energy detection or prior information related to the environment, at the cost of extra processing. It is also worth noting that the phase term in \eqref{eq:radarChannel} is a function of target's location, which plays a major role for increasing the accuracy of range estimation as discussed in Section~\ref{subsec:CRLB}.

%
%

\section{Index Modulation with Circularly-Shifted Chirps}
\label{sec:scheme}
\begin{figure*}
	\centering
	{\includegraphics[width =6.7in]{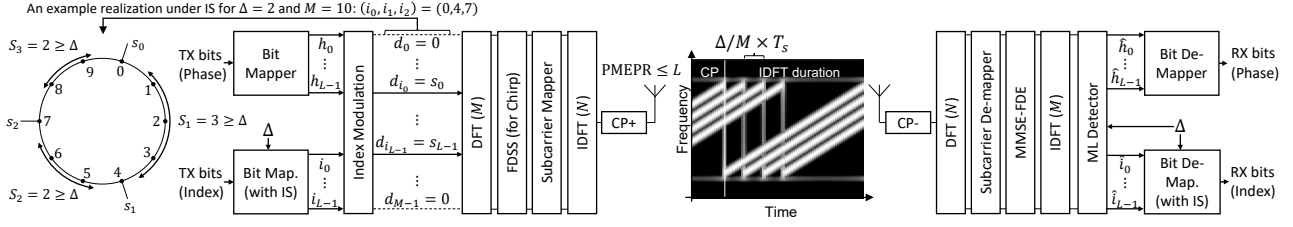}
	} 
\if\IEEEsubmission1
\vspace{-3mm}
\fi
	\caption{Transmitter and receiver for the proposed scheme and an example realization under IS for $\separationValue=2$ and $\numberOfShifts=10$.} 
	\label{fig:txrxa}
\if\IEEEsubmission1
\vspace{-3mm}
\fi
\end{figure*}
At the transmitter,  $\numberofTxBits$  information bits are first split into two groups: $\numofSelectorBits$  bits for choosing $\numberofIndices$ different \acp{CSC} from the set  $\chirpSet=\{\basisFunction[\indexDataSymbolTX]|\indexDataSymbolTX=0,1,\mydots,\numberOfShifts-1\}$  and $\numofModulationBits$ bits for $\numberofIndices$ different $\numberOfPointsForPSK$-\ac{PSK} symbols that are multiplied with the chosen \acp{CSC}.  
Let $\indexSequence=(\selectedChirpIndex[0],\selectedChirpIndex[1],\mydots,\selectedChirpIndex[\numberofIndices-1])$ for $0\le\selectedChirpIndex[\chirpm]<\selectedChirpIndex[\chirpn]<\numberOfShifts$ and $\chirpm<\chirpn$ and $\pskSerialSequence=(\symbolPSKserial[0],\symbolPSKserial[1],\mydots,\symbolPSKserial[\numberofIndices-1])$ for   $\symbolPSKserial[\indexChirp]\in\integers_\numberOfPointsForPSK$ denote the indices of the chosen \acp{CSC} and the integers to be mapped to $\numberOfPointsForPSK$-\ac{PSK} symbols, respectively.  Without loss of generality, the $\indexChirp$th $\numberOfPointsForPSK$-\ac{PSK} symbol can be calculated as $\symbolPSK[\indexChirp]=\constante^{\constantj2\pi\symbolPSKserial[\indexChirp]/\numberOfPointsForPSK}$  for $\indexChirp\in\{0,1,\mydots,\numberofIndices-1\}$. The baseband signal for the \ac{CSC-IM} 
can be written as
	$\transmittedSignal[\timeVar] = \frac{1}{\sqrt{\numberofIndices}} \sum_{\indexChirp=0}^{\numberofIndices-1} \symbolPSK[\indexChirp]\basisFunction[{\selectedChirpIndex[\indexChirp]}]$.
Therefore, by using \eqref{eq:chirpWave}, the discrete-time signal can be expressed  as
 \begin{align}
\transmittedSignalDiscrete[\indexSample]=&\sum_{\indexSubcarrier=\lowerFrequency}^{\upperFrequency}{\FDSScoef[\indexSubcarrier]{\frac{1}{\sqrt{\numberOfShifts}}\sum_{\indexDataSymbolTX=0}^{\numberOfShifts-1} \dataSymbols[\indexDataSymbolTX] 
			\constante^{-\constantj2\pi \indexSubcarrier \frac{\indexDataSymbolTX}{\numberOfShifts}}}}
	\constante^{\constantj2\pi \indexSubcarrier \frac{\indexSample}{\idftSize}}~,
\label{eq:chirpWaveIM}
\end{align}
where  $\FDSScoef[\indexSubcarrier]=\sqrt{\numberOfShifts}\fourierSeries[\indexSubcarrier]/ \sqrt{\sum_{\indexSubcarrier=\lowerFrequency}^{\upperFrequency}{|\fourierSeries[\indexSubcarrier]|^2}}$ is the $\indexSubcarrier$th normalized \ac{FDSS} coefficient\footnote{The normalization is applied as the coefficients are truncated in \eqref{eq:basisDecompose}.} and $\dataSymbols[{\selectedChirpIndex[\indexChirp]}]=\sqrt{\symbolEnergy}\times\symbolPSK[\indexChirp]$ for $\indexChirp=0,1,\mydots,\numberofIndices-1$, otherwise 0, and $\symbolEnergy=\numberOfShifts/\numberofIndices$ is the symbol energy.


In \figurename~\ref{fig:txrxa}, the transmitter and receiver block diagrams are given  for   \ac{CSC-IM}. First, the modulation symbols are obtained based on $\indexSequence$ and $\pskSerialSequence$. An $\numberOfShifts$-point DFT of the modulation symbols is then calculated. After the resulting sequence is shaped in the frequency domain with an \ac{FDSS} for generating \acp{CSC}, the shaped sequence is mapped to the \ac{OFDM} subcarriers. The
discrete-time signal in \eqref{eq:chirpWaveIM} is obtained by applying an $\idftSize$-point \ac{IDFT} to the mapped sequence and prepending a \ac{CP}  with the duration of $\CPDuration=\CPSize\Tsample$ to the signal, where $\CPSize$ is the number of samples in the \ac{CP} duration. 
At the \ac{RXc}, after removing the \ac{CP} and applying an $\idftSize$-point DFT to the received signal, the signal in the frequency domain  can be expressed as
\begin{align}
	\receivedSignalDiscreteFrequency[\indexSubcarrier] =\frac{\channelfreqresponseC[\indexSubcarrier]\FDSScoef[\indexSubcarrier]}{\sqrt{\numberOfShifts}} {{{\sum_{\indexDataSymbolTX=0}^{\numberOfShifts-1} \dataSymbols[\indexDataSymbolTX] 
				\constante^{-\constantj2\pi \indexSubcarrier \frac{\indexDataSymbolTX}{\numberOfShifts}}}}
		}+\noiseDiscrete[\indexSubcarrier]~,
	\label{eq:rxFreqSymbols}
\end{align}
where $\noiseDiscrete[\indexSubcarrier]$ is zero-mean \ac{AWGN} with the variance $\noiseVariance$ and $\channelfreqresponseC[\indexSubcarrier]$ is the \ac{CFR} between \ac{TX} and \ac{RXc} for the $\indexSubcarrier$th subcarrier for $\indexSubcarrier\in\{\lowerFrequency,\lowerFrequency+1,\mydots,\upperFrequency\}$. In this study, we consider  \ac{LMMSE} \ac{FDE} followed by an $\numberOfShifts$-point \ac{IDFT} to obtain the modulation symbols,  which can be given by
\begin{align}
	\dataSymbolAfterIDFTspread[\indexDataSymbolRX]=&\frac{1}{\sqrt{\numberOfShifts}} {{\sum_{\indexSubcarrier=\lowerFrequency}^{\upperFrequency}\underbrace{\frac{\channelfreqresponseC[\indexSubcarrier]^*\FDSScoef[\indexSubcarrier]^*}{|\channelfreqresponseC[\indexSubcarrier]\FDSScoef[\indexSubcarrier]|^2+\noiseVariance} \receivedSignalDiscreteFrequency[\indexSubcarrier]}_{\text{LMMSE-FDE}}\constante^{\constantj2\pi (\indexSubcarrier-\lowerFrequency) \frac{\indexDataSymbolRX}{\numberOfShifts}}}}~,
	\label{eq:mmse}
\end{align}
for $\indexDataSymbolRX\in\{1,2,\mydots,\numberOfShifts-1\}$. Note that \eqref{eq:mmse} explicitly shows that an equalizer needs to be employed  even in \ac{AWGN} channel (i.e., $\channelfreqresponseC[\indexSubcarrier]=1$) for \ac{CSC-IM} due to the \ac{FDSS} coefficients for \acp{CSC}. 

Without any constraint on the indices, the \ac{ML} detector for $\indexSequence$ and  $\pskSerialSequence$ can be expressed as
\begin{align}
    \{ \indexSequenceDetect, \pskSerialSequenceDetect\} = \arg\max_{\substack{\selectedChirpIndexDomain[\indexChirp]\in\{0,\mydots,\numberOfShifts-1\}\\ 		 	    
    		\selectedChirpIndexDomain[\chirpm]<\selectedChirpIndexDomain[\chirpn]  \textrm{ for }  \chirpm<\chirpn 
		\\ \symbolPSKserialdomain[\indexChirp]\in\integers_\numberOfPointsForPSK }
    		} \Re\left\{ \sum_{\indexChirp=0}^{\numberofIndices -1}{\dataSymbolAfterIDFTspread[{\selectedChirpIndexDomain[\indexChirp]}]      
    	 \constante^{-\constantj2\pi\symbolPSKserialdomain[\indexChirp]/\numberOfPointsForPSK}}\right\}~,
    \label{eq:mldet}
\end{align}
where $\indexSequenceDetect\triangleq(\selectedChirpIndexDetect[0],\selectedChirpIndexDetect[1],\mydots,\selectedChirpIndexDetect[\numberofIndices-1])$ and $\pskSerialSequenceDetect\triangleq(\symbolPSKserialdetect[0],\symbolPSKserialdetect[1],\mydots,\symbolPSKserialdetect[\numberofIndices-1])$  are the detected chirp  and the \ac{PSK} symbol indices, respectively. Thus, a low-complexity \ac{ML} detector can be implemented by evaluating each  $\datasymbolestimate[\indexDataSymbolRX][\indexPSK]\triangleq\Re\{ \dataSymbolAfterIDFTspread[\indexDataSymbolRX]\constante^{-\constantj2\pi\indexPSK/\numberOfPointsForPSK}\}$ for $\indexDataSymbolRX\in\{0,1,\mydots,\numberOfShifts-1\}$  and $\indexPSK\in\integers_\numberOfPointsForPSK$ and choosing $\numberofIndices$ indices and the corresponding $\indexPSK$ values that maximize $\datasymbolestimate[\indexDataSymbolRX][\indexPSK]$  \cite{Gao_2018}. 

Let $\separationDis[\indexSeparation]\in\integersNonnegativeSet$ denote the number of integers between two adjacent indices in a circular manner for $\indexSeparation\in\{1,\mydots,\numberofIndices\}$ as
\begin{align}
	\separationDis[\indexSeparation]\triangleq\begin{cases}
			\selectedChirpIndex[{\indexSeparation}]-\selectedChirpIndex[{\indexSeparation-1}]-1, & 1\le\indexSeparation<\numberofIndices\\\numberOfShifts-1-\selectedChirpIndex[{\numberofIndices-1}]+\selectedChirpIndex[{0}], & \indexSeparation=\numberofIndices 
	\end{cases}.
	\label{eq:sepDef}
\end{align}
With the \ac{IS},  we introduce constraints on the indices  such that  $\separationDis[\indexSeparation]\ge\separationValue$ for $\indexSeparation\in\{1,\mydots,\numberofIndices\}$ and $\separationValue\in\integersNonnegativeSet$  to improve the range estimation accuracy by ensuring a low \ac{AC} zone in the case of simultaneous \acp{CSC} transmissions. A realization under \ac{IS} for $\separationValue=2$ and $\numberOfShifts=10$ is given in \figurename~\ref{fig:txrxa}, where $(\selectedChirpIndex[{0}],\selectedChirpIndex[{1}],\selectedChirpIndex[{2}])=(0,4,7)$ and $(\separationDis[{1}],\separationDis[{2}],\separationDis[{3}])=(3,2,2)$. We discuss the \ac{IS} and its impact on \ac{TX} and \ac{RXc} in Section~\ref{subsec:IS} in detail.



\subsection{Theoretical error-rate analysis}
\label{subsubsec:union} 
Consider the case where  the \ac{FDSS} is not utilized, i.e., \ac{DFT-s-OFDM-IM}. Let ${\eventofeeror[\indexDataSymbolRX][\indexPSK][\indexSequence][\pskSerialSequence] }$ denote the event where $\datasymbolestimate[\indexDataSymbolRX][\indexPSK]$ is  larger than or equal to  at least one of the elements of $\{\datasymbolestimate[\selectedChirpIndex[\indexChirp]][{\symbolPSKserial[\indexChirp]}]\}$. Hence, based on De Morgan’s law, $\probabFunction[{\eventofeeror[\indexDataSymbolRX][\indexPSK][\indexSequence][\pskSerialSequence] }]$ can be obtained as 
\if \IEEEsubmission 0
\begin{align}
    \probabFunction[{\eventofeeror[\indexDataSymbolRX][\indexPSK][\indexSequence][\pskSerialSequence] }]&=\probabFunction[\bigcup_{\indexChirp=0}^{\numberofIndices -1}{\eventofeeror[\indexDataSymbolRX][\indexPSK][\selectedChirpIndex[\indexChirp]][{\symbolPSKserial[\indexChirp]}]}] =1-\probabFunction[\bigcap_{\indexChirp=0}^{\numberofIndices -1}{\eventofeeror[\indexDataSymbolRX][\indexPSK][\selectedChirpIndex[\indexChirp]][{\symbolPSKserial[\indexChirp]}]^{\rm c}}]\nonumber\\
    &=1-{\prod_{\indexChirp=0}^{\numberofIndices-1}\left(1-\probabFunction[{\eventofeeror[\indexDataSymbolRX][\indexPSK][\selectedChirpIndex[\indexChirp]][{\symbolPSKserial[\indexChirp]}]}]\right)}~.
    \label{eq:event}
\end{align}
\else
\begin{align}
	\probabFunction[{\eventofeeror[\indexDataSymbolRX][\indexPSK][\indexSequence][\pskSerialSequence] }]&=\probabFunction[\bigcup_{\indexChirp=0}^{\numberofIndices -1}{\eventofeeror[\indexDataSymbolRX][\indexPSK][\selectedChirpIndex[\indexChirp]][{\symbolPSKserial[\indexChirp]}]}] =1-\probabFunction[\bigcap_{\indexChirp=0}^{\numberofIndices -1}{\eventofeeror[\indexDataSymbolRX][\indexPSK][\selectedChirpIndex[\indexChirp]][{\symbolPSKserial[\indexChirp]}]^{\rm c}}]=1-{\prod_{\indexChirp=0}^{\numberofIndices-1}\left(1-\probabFunction[{\eventofeeror[\indexDataSymbolRX][\indexPSK][\selectedChirpIndex[\indexChirp]][{\symbolPSKserial[\indexChirp]}]}]\right)}~.
	\label{eq:event}
\end{align}
\fi
A block error occurs 
if $\indexDataSymbolRX$ is not an element of $\indexSequence$ or $\indexPSK\neq\symbolPSKserial[\indexChirp]$ for $\selectedChirpIndex[\indexChirp]=\indexDataSymbolRX$. 
The probability of block error can then be expressed as
\begin{align}
  \probabiltyError
  &=\probabFunction[{\underbrace{\bigcup_{\indexDataSymbolRX=0}^{\numberOfShifts-1} \bigcup_{\indexPSK=0}^{\numberOfPointsForPSK-1}}_{\substack{(\indexDataSymbolRX,\indexPSK) \neq (\selectedChirpIndex[\indexChirp],\symbolPSKserial[\indexChirp])\\ \indexChirp\in\{{0,\mydots {\numberofIndices-1}\}}}}{\eventofeeror[\indexDataSymbolRX][\indexPSK][\indexSequence][\pskSerialSequence] }}]\leq \underbrace{\sum_{\indexDataSymbolRX=0}^{\numberOfShifts-1} \sum_{\indexPSK=0}^{\numberOfPointsForPSK-1}}_{\substack{(\indexDataSymbolRX,\indexPSK) \neq (\selectedChirpIndex[\indexChirp],\symbolPSKserial[\indexChirp])\\ \indexChirp\in\{{0,\mydots {\numberofIndices-1}\}}}}\probabFunction[{\eventofeeror[\indexDataSymbolRX][\indexPSK][\indexSequence][\pskSerialSequence] }]~.
  \label{eq:unionboundPre}
  \end{align}
 By using \eqref{eq:event}, \eqref{eq:unionboundPre} can be given by
 \if \IEEEsubmission 0
	  \begin{align}
	    \probabiltyError&\leq \underbrace{\sum_{\indexDataSymbolRX=0}^{\numberOfShifts-1} \sum_{\indexPSK=0}^{\numberOfPointsForPSK-1}}_{\substack{(\indexDataSymbolRX,\indexPSK) \neq (\selectedChirpIndex[\indexChirp],\symbolPSKserial[\indexChirp])\\ \indexChirp\in\{{0,\mydots {\numberofIndices-1}\}}}}{\left(1-{\prod_{\indexChirp=0}^{\numberofIndices-1}\left(1-\probabFunction[{\eventofeeror[\indexDataSymbolRX][\indexPSK][\selectedChirpIndex[\indexChirp]][{\symbolPSKserial[\indexChirp]}]}]\right)}\right)}\nonumber\\
	    &=\underbrace{\sum_{\indexDataSymbolRX=0}^{\numberOfShifts-1} \sum_{\indexPSK=0}^{\numberOfPointsForPSK-1}}_{\indexDataSymbolRX\notin\indexSequence}{\left(1-{\prod_{\indexChirp=0}^{\numberofIndices-1}\left(1-\probabFunction[{\eventofeeror[\indexDataSymbolRX][\indexPSK][\selectedChirpIndex[\indexChirp]][{\symbolPSKserial[\indexChirp]}]}]\right)}\right)}
	   	\nonumber\\ &
	    +\sum_{\indexChirp=0}^{\numberofIndices -1} \underbrace{\sum_{\indexPSK=0}^{\numberOfPointsForPSK-1}}_{\indexPSK\neq\symbolPSKserial[\indexChirp]}{\left(1-{\prod_{\indexChirp=0}^{\numberofIndices-1}\left(1-\probabFunction[{\eventofeeror[\indexDataSymbolRX][\indexPSK][\selectedChirpIndex[\indexChirp]][{\symbolPSKserial[\indexChirp]}]}]\right)}\right)}~.
	    \label{eq:unionbound}
	\end{align}
\else
	  \begin{align}
	\probabiltyError&\leq \underbrace{\sum_{\indexDataSymbolRX=0}^{\numberOfShifts-1} \sum_{\indexPSK=0}^{\numberOfPointsForPSK-1}}_{\substack{(\indexDataSymbolRX,\indexPSK) \neq (\selectedChirpIndex[\indexChirp],\symbolPSKserial[\indexChirp])\\ \indexChirp\in\{{0,\mydots {\numberofIndices-1}\}}}}{1-{\left(\prod_{\indexChirp=0}^{\numberofIndices-1}\left(1-\probabFunction[{\eventofeeror[\indexDataSymbolRX][\indexPSK][\selectedChirpIndex[\indexChirp]][{\symbolPSKserial[\indexChirp]}]}]\right)\right)}}\nonumber\\
	&=\underbrace{\sum_{\indexDataSymbolRX=0}^{\numberOfShifts-1} \sum_{\indexPSK=0}^{\numberOfPointsForPSK-1}}_{\indexDataSymbolRX\notin\indexSequence}{1-{\left(\prod_{\indexChirp=0}^{\numberofIndices-1}\left(1-\probabFunction[{\eventofeeror[\indexDataSymbolRX][\indexPSK][\selectedChirpIndex[\indexChirp]][{\symbolPSKserial[\indexChirp]}]}]\right)\right)}}
	+\sum_{\indexChirp=0}^{\numberofIndices -1} \underbrace{\sum_{\indexPSK=0}^{\numberOfPointsForPSK-1}}_{\indexPSK\neq\symbolPSKserial[\indexChirp]}{1-{\left(\prod_{\indexChirp=0}^{\numberofIndices-1}\left(1-\probabFunction[{\eventofeeror[\indexDataSymbolRX][\indexPSK][\selectedChirpIndex[\indexChirp]][{\symbolPSKserial[\indexChirp]}]}]\right)\right)}}~.
	\label{eq:unionbound}
	\end{align}
\fi
The Euclidean distance between $\dataSymbols[{\selectedChirpIndex[\indexChirp]}]$ and $ \dataSymbols[{\indexDataSymbolRX\neq\selectedChirpIndex[\indexChirp]}]$ is fixed and can be calculated  as $\distanceminindex=\sqrt{2\symbolEnergy}$. The minimum Euclidean distance between two \ac{PSK} constellation points for the same index is $\distanceminpsk=2\sqrt{\symbolEnergy}{\sin(\frac{\pi}{\numberOfPointsForPSK})}$. Under the coherent detection and by using the symbol-error rate for  $\numberOfPointsForPSK$-\ac{PSK} \cite{proakisfundamentals}, this implies that the right-hand side of \eqref{eq:unionbound} can be re-written as
\if \IEEEsubmission 0
\begin{align}
	\probabiltyError\leq &\probabiltyunionbound[\numberofIndices]\triangleq
	(\numberOfShifts-\numberofIndices)\numberOfPointsForPSK\left(1-\left(1-\qfunc[\frac{\distanceminindex}{\sqrt{2\noisePower}}]\right)^{\numberofIndices}\right)\nonumber\\
	&~~~~~~~~~~~~~~~~~~~~~~~~+\numberofIndices\left(1-\left(1-\probabiltyErrorPSK\right)^{\numberofIndices}\right)~,
	\label{eq:probabUnionBound}
\end{align}
\else
\begin{align}
	\probabiltyError\leq\probabiltyunionbound[\numberofIndices]\triangleq
	(\numberOfShifts-\numberofIndices)\numberOfPointsForPSK\left(1-\left(1-\qfunc[\frac{\distanceminindex}{\sqrt{2\noisePower}}]\right)^{\numberofIndices}\right)+\numberofIndices\left(1-\left(1-\probabiltyErrorPSK\right)^{\numberofIndices}\right)~,
	\label{eq:probabUnionBound}
\end{align}
\fi
where 
\begin{align}
\probabiltyErrorPSK= \begin{cases}
	2\qfunc[\frac{\distanceminpsk}{\sqrt{2\noisePower}}], & \numberOfPointsForPSK\ge4\\
	\qfunc[\frac{\distanceminpsk}{\sqrt{2\noisePower}}], & \numberOfPointsForPSK=2\\	
	0, & \numberOfPointsForPSK=1\\
\end{cases}~,
\end{align}
$\qfunc[\cdot]$ is the Q-function, and $\probabiltyunionbound[\numberofIndices]$ is the \ac{UB} of the probability of error for any $\numberofIndices\in\{1,\mydots,\numberOfShifts\}$ and $\numberOfPointsForPSK\ge1$ for \ac{DFT-s-OFDM-IM} in  \ac{AWGN} channel. 

Now consider the case where the \ac{FDSS} for \acp{CSC} is included, i.e., \ac{CSC-IM}. Since we use a single-tap \ac{LMMSE}-\ac{FDE} followed by an $\numberOfShifts$-point \ac{IDFT}  as in \eqref{eq:mmse},  \ac{ISI} occurs and the noise on the received modulation symbols becomes correlated. By utilizing the derivation in \cite{Nisar_2007} for a typical \ac{DFT-s-OFDM} transmission  with \ac{iid} modulation symbols under a fading channel, the \ac{SNR} for the received modulation symbols after the equalization can be obtained as $\SNRpost = 1/({\sqrt{1/\signalFactor} -1})$, where $\signalFactor=\left(1/\numberOfShifts\sum_{\indexSubcarrier=\lowerFrequency}^{\upperFrequency} {|\FDSScoef[\indexSubcarrier]|^2}/({|\FDSScoef[\indexSubcarrier]|^2+\noiseVariance})\right)^2$. Assuming that the noise on the modulation symbols is uncorrelated and follows Gaussian distribution, the bound in \eqref{eq:probabUnionBound} can be calculated by using $\noisePower=1/\SNRpost$. It is worth noting that these assumptions are weak as long as the parameters do not cause ill-conditioned operations (e.g., $\numberOfOccupiedSubcarriers\ll\numberOfShifts$) and \eqref{eq:probabUnionBound} is fairly accurate for typical \ac{FDSS} choices, e.g., \eqref{eq:fresnekfcn}, as shown in Section~\ref{sec:numresults}.

\subsection{Trade-off between PMEPR and Spectral Efficiency }
\label{subsec:connection}
\begin{figure}[t]
	\centering
	{\includegraphics[width =3.3in]{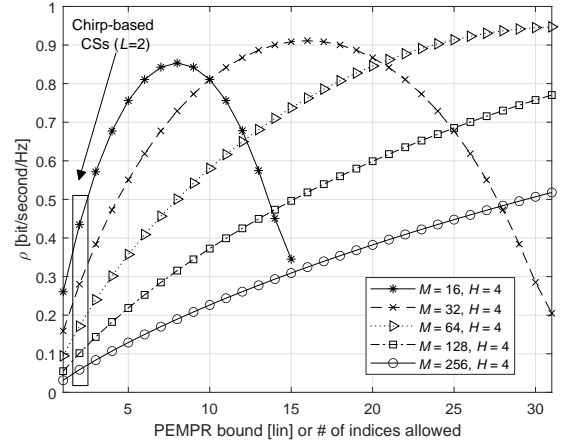}
	} 
	\caption{Trade-off between the \ac{SE} and the maximum \ac{PMEPR} for \ac{CSC-IM}.}
	\label{fig:tradeoff}
\if\IEEEsubmission1
\vspace{-3mm}
\fi
\end{figure}
The \ac{CSC-IM} allows $\numberofTxBits=\numofSelectorBits+\numofModulationBits$  information bits to be transmitted, where $\numofSelectorBits=\floor{\log_{2}\left({\binom{\numberOfShifts}{\numberofIndices}}\right)}$ and $\numofModulationBits= \numberofIndices\log_{2}(\numberOfPointsForPSK)$ since $\numberofIndices$ indices can be chosen from $\numberOfShifts$ indices in ${\binom{\numberOfShifts}{\numberofIndices}}$ different ways and  $\numberofIndices$ $\numberOfPointsForPSK$-\ac{PSK} symbols  are utilized. Hence, the \ac{SE} of the scheme can be calculated as
$ \spectralEfficiency = {\floor{\log_{2}({{\numberOfShifts \choose \numberofIndices}\times \numberOfPointsForPSK^\numberofIndices})}}/{\numberOfShifts}$ bit/second/Hz. Therefore, the \ac{SE} increases with $\numberofIndices\le\floor{\numberOfShifts/2}$.
On the other hand, since chirps are constant-envelope signals, i.e., $|\basisFunction[\indexDataSymbolRX]|=1$, the maximum amplitude of the superposition of $\numberofIndices$ chirps is less than or equal to  $\numberofIndices$, i.e., the instantaneous power is maximum $\numberofIndices^2$. As the mean power is $\numberofIndices$ in the presence of random \ac{PSK} symbols, the \ac{PMEPR} of the signal with \ac{CSC-IM}  is always less than or equal to $\numberofIndices$. Therefore, the proposed scheme leads to a trade-off between maximum \ac{PMEPR} and \ac{SE} as quantified in \figurename~\ref{fig:tradeoff} by sweeping $\numberofIndices$ for a given $\numberofIndices$ and for a \ac{QPSK} constellation, i.e., $\numberOfPointsForPSK=4$, and $\numberOfShifts\in\{16,32,64,128,256\}$.  While the \ac{SE} of \ac{CSC-IM} increases with a larger $\numberofIndices\le\floor{\numberOfShifts/2}$, the maximum \ac{PMEPR} is always less than or equal to $10\log_{10}{\numberofIndices}$ dB. 
This trade-off can be helpful to identify the maximum number of \acp{CSC} and the corresponding \ac{SE} for a given power back-off. For example, if the tolerable input power back-off is 7 dB, the maximum number of active \acp{CSC} is $5$ and  the maximum \ac{SE} is $0.35$~bits/(s.Hz) for $\numberOfShifts=64$ based on \figurename~\ref{fig:tradeoff}.

As a special case, the proposed scheme for $\numberofIndices=2$ reveals that non-trivial \acp{CS} can be generated from chirps as follows:
\begin{theorem}
	\label{th:golaypair}
	Let $\seqx[\timeVar]$ and $\seqy[\timeVar]$ be the signals given by
	\begin{align}
		\seqx[\timeVar] = \dataSymbols[\chirpm]\constante^{\constantj\chirpphase[\chirpm][\timeVar]} + \dataSymbols[\chirpn]\constante^{\constantj\chirpphase[\chirpn][\timeVar]}~,
		\label{eq:Golay1}
		\\
		\seqy[\timeVar] = \dataSymbols[\chirpm]\constante^{\constantj\chirpphase[\chirpm][\timeVar]} - \dataSymbols[\chirpn]\constante^{\constantj\chirpphase[\chirpn][\timeVar]}~,
		\label{eq:Golay2}
	\end{align}
	for $\dataSymbols[\chirpm],\dataSymbols[\chirpn]\in\complexNumbers$ and $|\dataSymbols[\chirpm]|=|\dataSymbols[\chirpn]|=1$. The Fourier coefficients of $\seqx[\timeVar]$ and $\seqy[\timeVar]$ form a \ac{GCP}.
\end{theorem}
\begin{proof}
	By the definition of a \ac{GCP}, we need to show that $|\seqx[\timeVar]|^2+|\seqy[\timeVar]|^2$ is constant:
	\if \IEEEsubmission 0
		\begin{align}
			|\seqx[\timeVar]|^2 
			=&|\dataSymbols[\chirpm]|^2+|\dataSymbols[\chirpn]|^2\nonumber\\
			&+\dataSymbols[\chirpm]\dataSymbols[\chirpn]^*\constante^{\constantj(\chirpphase[\chirpm][\timeVar]-\chirpphase[\chirpn][\timeVar])}+\dataSymbols[\chirpm]^*\dataSymbols[\chirpn]\constante^{-\constantj(\chirpphase[\chirpm][\timeVar]-\chirpphase[\chirpn][\timeVar])}~.\nonumber
		\end{align}
	\else
		\begin{align}
		|\seqx[\timeVar]|^2=|\dataSymbols[\chirpm]|^2+|\dataSymbols[\chirpn]|^2+\dataSymbols[\chirpm]\dataSymbols[\chirpn]^*\constante^{\constantj(\chirpphase[\chirpm][\timeVar]-\chirpphase[\chirpn][\timeVar])}+\dataSymbols[\chirpm]^*\dataSymbols[\chirpn]\constante^{-\constantj(\chirpphase[\chirpm][\timeVar]-\chirpphase[\chirpn][\timeVar])}~.\nonumber
	\end{align}
	\fi	
	Similarly,
	\if \IEEEsubmission 0
	\begin{align}
		|\seqy[\timeVar]|^2
		=&|\dataSymbols[\chirpm]|^2+|\dataSymbols[\chirpn]|^2\nonumber\\
		&-\dataSymbols[\chirpm]\dataSymbols[\chirpn]^*\constante^{\constantj(\chirpphase[\chirpm][\timeVar]-\chirpphase[\chirpn][\timeVar])}-\dataSymbols[\chirpm]^*\dataSymbols[\chirpn]\constante^{-\constantj(\chirpphase[\chirpm][\timeVar]-\chirpphase[\chirpn][\timeVar])}~.\nonumber
	\end{align}
	\else
	\begin{align}
	|\seqy[\timeVar]|^2
	=|\dataSymbols[\chirpm]|^2+|\dataSymbols[\chirpn]|^2-\dataSymbols[\chirpm]\dataSymbols[\chirpn]^*\constante^{\constantj(\chirpphase[\chirpm][\timeVar]-\chirpphase[\chirpn][\timeVar])}-\dataSymbols[\chirpm]^*\dataSymbols[\chirpn]\constante^{-\constantj(\chirpphase[\chirpm][\timeVar]-\chirpphase[\chirpn][\timeVar])}~.\nonumber
\end{align}
	\fi
	Therefore, $|\seqx[\timeVar]|^2+|\seqy[\timeVar]|^2
	=2\times(|\dataSymbols[\chirpm]|^2+|\dataSymbols[\chirpn]|^2)
	=4$,
	which implies that $\fourier[{\seqx[{\timeVar}]}]$ and  $\fourier[{\seqy[{\timeVar}]}]$ form a \ac{GCP}.
\end{proof}
Theorem~\ref{th:golaypair} indicates that the Fourier coefficients of a linear combination of the frequency responses of two constant-envelope chirps result in a \ac{CS}. Hence, based on \eqref{eq:fresnekfcn} and \eqref{eq:besselfcn}, Fresnel integrals and Bessel functions can be useful for generating \acp{CS}, which have not been reported in the literature to the best of our knowledge. 
\begin{example}
	\rm 
	Assume that $\seqx[\timeVar]$ and $\seqy[\timeVar]$ are  linear combinations of two circularly-shifted versions of  a band-limited sinusoidal chirp. By using \eqref{eq:basisDecompose} and \eqref{eq:besselfcn}, the Fourier coefficients of $\seqx[\timeVar]$ and $\seqy[\timeVar]$ are obtained as
	\begin{align}
		\eleGa[\indexSubcarrier] &= \dataSymbols[\chirpm]\besselFunctionFirstKind[\indexSubcarrier][\frac{\numberOfOccupiedSubcarriers}{2}]\constante^{-\constantj2\pi\indexSubcarrier\frac{\amountOfShift[\chirpm]}{\symbolDuration}}+\dataSymbols[\chirpn]\besselFunctionFirstKind[\indexSubcarrier][\frac{\numberOfOccupiedSubcarriers}{2}]\constante^{-\constantj2\pi\indexSubcarrier\frac{\amountOfShift[\chirpn]}{\symbolDuration}}~, 
		\label{eq:CSbesselA}
		\\
		\eleGb[\indexSubcarrier] &= \dataSymbols[\chirpm]\besselFunctionFirstKind[\indexSubcarrier][\frac{\numberOfOccupiedSubcarriers}{2}]\constante^{-\constantj2\pi\indexSubcarrier\frac{\amountOfShift[\chirpm]}{\symbolDuration}}-\dataSymbols[\chirpn]\besselFunctionFirstKind[\indexSubcarrier][\frac{\numberOfOccupiedSubcarriers}{2}]\constante^{-\constantj2\pi\indexSubcarrier\frac{\amountOfShift[\chirpn]}{\symbolDuration}}\label{eq:CSbesselB},
	\end{align}
	respectively. Based on Theorem~\ref{th:golaypair}, $(\eleGa[{\indexEleOfSeq}])_{i=-\infty}^{\infty}$ and $(\eleGb[{\indexEleOfSeq}])_{i=-\infty}^{\infty}$ form a \ac{GCP}. Since the sinusoidal chirps are band-limited signals, the amplitude of a Fourier coefficient approaches to zero for $|\indexEleOfSeq|\ge\numberOfOccupiedSubcarriers/2$. Therefore, $(\eleGa[{\indexEleOfSeq}])_{i=\lowerFrequency}^{\upperFrequency}$ and $(\eleGb[{\indexEleOfSeq}])_{i=\lowerFrequency}^{\upperFrequency}$ are approximately \ac{GCP}. Note that if the sinusoidal chirps are replaced by the linear chirps, by using \eqref{eq:fresnekfcn},  the Fourier coefficients of $\seqx[\timeVar]$ and $\seqy[\timeVar]$ can be calculated as
	\begin{align}
		\eleGa[\indexSubcarrier] =& \dataSymbols[\chirpm]
		\linearCoef(\fresnelC[{\linearXone}] + \fresnelC[{\linearXtwo}] + \constantj\fresnelS[{\linearXone}] + \constantj\fresnelS[{\linearXtwo}])
		\constante^{-\constantj2\pi\indexSubcarrier\frac{\amountOfShift[\chirpm]}{\symbolDuration}}\nonumber\\ &+\dataSymbols[\chirpn]\linearCoef(\fresnelC[{\linearXone}] + \fresnelC[{\linearXtwo}] + \constantj\fresnelS[{\linearXone}] + \constantj\fresnelS[{\linearXtwo}])\constante^{-\constantj2\pi\indexSubcarrier\frac{\amountOfShift[\chirpn]}{\symbolDuration}}~,\nonumber
		\\
		\eleGb[\indexSubcarrier] =& \dataSymbols[\chirpm]
		\linearCoef(\fresnelC[{\linearXone}] + \fresnelC[{\linearXtwo}] + \constantj\fresnelS[{\linearXone}] + \constantj\fresnelS[{\linearXtwo}])
		\constante^{-\constantj2\pi\indexSubcarrier\frac{\amountOfShift[\chirpm]}{\symbolDuration}}\nonumber\\ &-\dataSymbols[\chirpn]\linearCoef(\fresnelC[{\linearXone}] + \fresnelC[{\linearXtwo}] + \constantj\fresnelS[{\linearXone}] + \constantj\fresnelS[{\linearXtwo}])\constante^{-\constantj2\pi\indexSubcarrier\frac{\amountOfShift[\chirpn]}{\symbolDuration}}~.\nonumber
	\end{align}
\end{example}

In \figurename~\ref{fig:correlationExample}, we exemplify a \ac{GCP} of length $\numberOfShifts=24$, synthesized through \eqref{eq:CSbesselA} and \eqref{eq:CSbesselB} for $\lowerFrequency=-11$, $\upperFrequency=12$,  ${\amountOfShift[\chirpm]}/{\symbolDuration}=0/24$, ${\amountOfShift[\chirpn]}/{\symbolDuration}=1/24$, and $\dataSymbols[\chirpm]=\dataSymbols[\chirpn]=1$. When $\numberOfOccupiedSubcarriers=24$, the sequences are truncated heavily. Therefore, it does not satisfy the definition of a \ac{GCP} given in Section~\ref{subsec:complementarySequence}. On the other hand, when the maximum frequency deviation is halved, $\OCB$ is $15$ for containing $99\%$  of the total integrated power of the  spectrum. Hence, $\numberOfShifts=24$ forms the chirps well and the resulting sequences form a \ac{GCP}. It is also worth noting that synthesized \acp{CS} are not unimodular sequences. Therefore, the mean power of \ac{OFDM} symbol changes although  instantaneous power is bounded.

\begin{figure}[t]
	\centering
	{\includegraphics[width =3.3in]{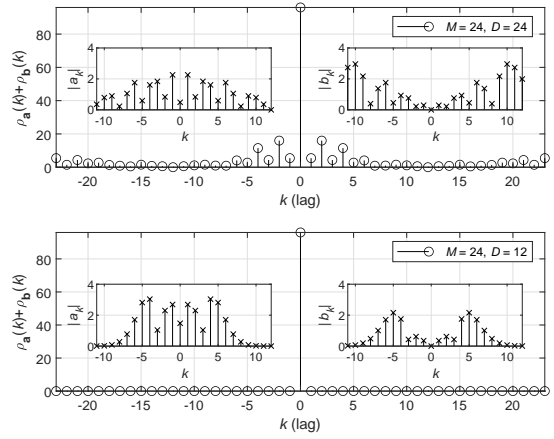}
	} 
	\caption{A \ac{GCP} of length $\numberOfShifts=24$ synthesized with sinusoidal chirps and the impact of truncation on the \ac{GCP} formation.
	}
\if\IEEEsubmission1
\vspace{-3mm}
\fi
	\label{fig:correlationExample}
\end{figure}
\begin{corollary} 
	\label{co:indexmodulation}
	Let $\chirpSet=\{\basisFunction[\indexDataSymbolTX]|\indexDataSymbolTX=0,1,\mydots,\numberOfShifts-1\}$ be a set of $\numberOfShifts$  \acp{CSC} of an  arbitrary band-limited function with the duration $\symbolDuration$ and $\OCB\le\numberOfShifts$. Without using the same \ac{CSC} twice, the total number of distinct \acp{CS} of length  $\numberOfShifts$ is ${\numberOfShifts \choose 2}\times \numberOfPointsForPSK^2$ for  $\dataSymbols[\chirpm],\dataSymbols[\chirpn]\in\{\constante^{\constantj 2\pi\indexPSK/\numberOfPointsForPSK }| \indexPSK=0,1,\mydots, \numberOfPointsForPSK-1\}$. 
\end{corollary} 

\begin{proof}
	There exist $\numberOfPointsForPSK^2$ combinations for $\{\dataSymbols[\chirpm],\dataSymbols[\chirpn]\}$ and $\constante^{\constantj\chirpphase[\chirpm][\timeVar]}$ and $\constante^{\constantj\chirpphase[\chirpn][\timeVar]}$ in \eqref{eq:Golay1} can be chosen in ${\numberOfShifts \choose 2}$ ways without using the same chirp.
	Thus, the total number of \acp{CS} is ${\numberOfShifts \choose 2}\times \numberOfPointsForPSK^2$ via Theorem~\ref{th:golaypair}. The \acp{CS} are distinct as $\basisFunction[\indexDataSymbolTX]$ are distinct for $\indexDataSymbolTX\in\{0,1,\mydots,\numberOfShifts-1\}$. Since the \ac{OCB} of $\basisFunction[\indexDataSymbolTX]$ is less than or equal to $\numberOfShifts/\symbolDuration$, the length of the synthesized \ac{CS} is $\numberOfShifts$ based on Nyquist's sampling theorem.
\end{proof}

Note that the extension of Theorem~\ref{th:golaypair} for $\numberofIndices>2$
is currently a difficult open problem. To address this issue,	using the coefficients of \acp{CSC} as seed \ac{GCP} sequences for the recursive \ac{CS} construction \cite{sahin_2020gm} is a potential direction that can be pursued.

\subsection{Practical Issues}
At the transmitter, mapping the information bits to a combination of $\numberofIndices$ indices (and vice versa for \ac{RXr}) may be a challenge. This can be addressed by constructing a bijective function from integers to the set of combinations via a combinatorial number system of degree $\numberofIndices$ \cite{basar_2013}, also called combinadics. Note that we investigate a generalization of combinadics in Section~\ref{subsubsec:bijection} to obtain a mapping rule under the proposed \ac{IS} to develop the encoder and decoder, which may also be used for the unconstrained case.

The choice of \ac{FDSS} is important for obtaining a low error rate and a low \ac{PMEPR}. In \cite{sahin_2020}, it was demonstrated that a flatter \ac{FDSS} improves the \ac{BER} performance for the receiver with a single-tap \ac{LMMSE}-\ac{FDE}. For example, a linear chirp offers a more flat \ac{FDSS} as compared to the one with a sinusoidal chirp. On the other hand, a linear chirp causes abrupt instantaneous frequency changes within the \ac{IDFT} duration. Therefore, it requires a much lower $\numberOfOccupiedSubcarriers$ as compared to the one for a sinusoidal chirp for a given $\numberOfShifts$ to form a \ac{CS}. This issue can distort the signals and cause a larger \ac{PMEPR} than the theoretical bound, as demonstrated in Section~\ref{sec:numresults}. Note that the \ac{PMEPR} can be reduced to the theoretical limit if \ac{FDSS} is allowed to be extended to the sidebands with some roll-off factor, i.e., a larger $\upperFrequency$ and a lower $\lowerFrequency$ in \eqref{eq:chirpWave}.

\subsection{Comparisons}
As compared to \ac{DFT-s-OFDM-IM}, the proposed scheme has a significant \ac{PMEPR} advantage since  the signal is spread in the time domain whereas \ac{DFT-s-OFDM} generates Dirichlet sinc pulses. The proposed scheme and \ac{OFDM-IM} have similar \ac{PMEPR} characteristics since \ac{OFDM-IM} also spreads the symbol energy in time\footnote{In this study, we consider a large number of subcarriers with few indices for both \ac{OFDM-IM} and DFT-s-OFDM-IM.}. However, the energy is also distributed  within the signal bandwidth with the proposed scheme. Thus, the proposed scheme allows a coherent receiver to exploit frequency diversity in frequency-selective channels naturally. On the other hand, \ac{OFDM-IM} receiver does not fully benefit from the frequency selectivity without an extra operation, e.g., repetitions or interleaving  \cite{Basar2015_CIIM}, at the transmitter.

In \cite{davis_1999}, a low \ac{PMEPR} coding scheme was proposed to generate $\numberOfPointsForPSK$-\ac{PSK} \ac{CS-RM}. This scheme synthesizes $\numberOfPointsForPSK^{\positiveInt+1}\times \positiveInt!/2$ \acp{CS}, where the length of each \ac{CS}  must be in the form of $2^\positiveInt$, where $\positiveInt\in\integersPositiveSet$. 
When a seed \ac{GCP} of length $N$ is utilized with this scheme,  it can be shown that $\numberOfPointsForPSK^{\positiveInt+1}\times \positiveInt!$ \acp{CS} of length $N\cdot 2^\positiveInt$ can be generated \cite{sahin_2020gm}. Therefore, the spectral efficiency of the schemes in  \cite{davis_1999} and \cite{sahin_2020gm} can be calculated as $\floor{\log_2 ({ \numberOfPointsForPSK^{\positiveInt+1}\times \positiveInt!/2})}/2^{\positiveInt}$ and $\floor{\log_2 ({ \numberOfPointsForPSK^{\positiveInt+1}\times \positiveInt!})}/(N2^{\positiveInt})$ bit/second/Hz, respectively. The differences between these schemes and the proposed scheme can be listed as follows: 1) The proposed scheme supports flexible bandwidth. For example, $\numberOfShifts$ can be an integer chosen as an integer multiple of 12 based on the resource allocation in 3GPP \ac{5G} \ac{NR} and \ac{4G} \ac{LTE}. 2) The schemes in \cite{davis_1999, sahin_2020gm} do not provide a trade-off between \ac{PMEPR} and spectral efficiency whereas  $\numberofIndices$ can be chosen for a higher \ac{SE} at a cost of high \ac{PMEPR} with our scheme. The \ac{PMEPR} is still theoretically limited. 3) The number of \acp{CS} is a function of a second-order coset term generated through permutations  in \cite{davis_1999} and \cite{sahin_2020gm}. However, designing a decoder for all possible permutations is not trivial \cite{Paterson2000decode, sahin_2020gm}. For a fixed coset, the decoder can be implemented through fast Hadamard transformation or recursive methods \cite{Schmidt2005}, but the \ac{SE} reduces to $\floor{\log_2 { \numberOfPointsForPSK^{m+1}}}$. Under this case, the \ac{SE} of the proposed scheme and the schemes in \cite{davis_1999} and \cite{sahin_2020gm} are similar while a simple decoder can be employed for the proposed method.

Note that  the \ac{SE} of the proposed scheme is low as compared to typical coding schemes such as \ac{LDPC} or polar codes. Although this appears as a disadvantage, there exist many communication scenarios (e.g., uplink control channels in \ac{5G} \ac{NR} \cite{Sahin_20twc}, \ac{IoT} networks) where the primary concern is reliability under low \ac{SNR} or a longer battery life, rather than a higher data rate. In addition, to exploit the demodulated data for improving radar functionality  for bi-static \ac{DFRC} scenarios,  the signal should be able to decoded at very low \ac{SNR}. For these scenarios, the proposed scheme provides a way of limiting \ac{PMEPR} without an optimization procedure at the transmitter while exploiting frequency selectivity and supporting radar functionality discussed next.

\section{Radar Functionality with CSC-IM}
\label{sec:radar}
At the \ac{RXr}, similar to \eqref{eq:rxFreqSymbols}, the received signal in the frequency domain  can be expressed as
\begin{align}
	\receivedSignalDiscreteFrequency[\indexSubcarrier] =
	 {\frac{\channelfreqresponse[\indexSubcarrier]\FDSScoef[\indexSubcarrier]}{\sqrt{\numberOfShifts}}{{\sum_{\indexDataSymbolTX=0}^{\numberOfShifts-1} \dataSymbols[\indexDataSymbolTX] 
				\constante^{-\constantj2\pi \indexSubcarrier \frac{\indexDataSymbolTX}{\numberOfShifts}}}}
	}+\noiseDiscrete[\indexSubcarrier]~,
	\label{eq:rxFreqSymbolsR}
\end{align}
where $\noiseDiscrete[\indexSubcarrier]$ is zero-mean \ac{AWGN} with the variance of $\noiseVariance$. The received symbols in \eqref{eq:rxFreqSymbolsR} can be re-expressed in the vector form as
\begin{align}
	\rxSymbolsVector= \underbrace{\diagonalMatrixFromVector[{ \fdssVector }]\diagonalMatrixFromVector[{  \DFTmtx[\numberOfShifts] \dataVector }]}_{\completeMatrix\triangleq\diagonalMatrixFromVector[{\symbolsVectorinFrequency}]} \channelFVector + \noiseVector~,
\end{align}
where $\rxSymbolsVector=[\receivedSignalDiscreteFrequency[\lowerFrequency],\mydots,\receivedSignalDiscreteFrequency[\upperFrequency] ]^{\rm T}$, $\fdssVector=[\FDSScoef[\lowerFrequency], \mydots,\FDSScoef[\upperFrequency] ]$, $\DFTmtx[\numberOfShifts]$ is the $\numberOfShifts$-point normalized DFT matrix,
$\dataVector=[\dataSymbols[\lowerFrequency],  \mydots,\dataSymbols[\upperFrequency] ]^{\rm T}$, 
 $\noiseVector=[\noiseDiscrete[\lowerFrequency], 
 \noiseDiscrete[\lowerFrequency+1], \mydots,\noiseDiscrete[\upperFrequency] ]^{\rm T}$,
 $\symbolsVectorinFrequency=[\symbolsInFrequency[\lowerFrequency], \mydots, \symbolsInFrequency[\upperFrequency]]^{\rm T}$ is the response of the waveform in the frequency, and $\channelFVector=[\channelfreqresponse[\lowerFrequency], \mydots,\channelfreqresponse[\upperFrequency]]^{\rm T}$. Based on \eqref{eq:radarChannel}, $\channelFVector$ can be expressed as
\begin{align}
\channelFVector = \delayMtx\amplitudeVector~,
\end{align}
where $\delayMtx=[\delayVector[{\timeArrival[1]}]~\delayVector[{\timeArrival[2]}]~\cdots~\delayVector[{\timeArrival[\numberoftargets]}]]\in\complexNumbers^{\numberOfShifts\times\numberoftargets}$ is the delay matrix and $\delayVector[{\timeArrival[\targetindex]}]=\constante^{-\constantj2\pi\fcarrier\timeArrival[\targetindex]}\times[ \constante^{-\constantj2\pi\lowerFrequency\frac{\timeArrival[\targetindex]}{\symbolDuration}},\cdots,\constante^{-\constantj2\pi\upperFrequency\frac{\timeArrival[\targetindex]}{\symbolDuration}}]^{\rm T}$, and $\amplitudeVector=[\reflectioncoefficient[1], \reflectioncoefficient[2], \mydots,\reflectioncoefficient[\numberoftargets] ]^{\rm T}$. For our \ac{DFRC} scenario, the sequences $\pskSequence$ and $\indexSequence$ are available at the \ac{RXr}. Therefore, the symbols on the subcarriers, i.e., $\symbolsVectorinFrequency$, can be used as reference symbols.  Hence, in \ac{AWGN} channel, the \ac{ML}-based delay estimation problem can be expressed as
\if \IEEEsubmission 0
	\begin{align}
	\{(\timeArrivalEst[\targetindex],\reflectioncoefficientEst[\targetindex])\}&=\arg\min_{
	\substack{ \{(\timeArrivalSlack[\targetindex], \reflectioncoefficientSlack[\targetindex])\} \\ \targetindex=1,\mydots,\numberoftargets}
	 } \lVert  \rxSymbolsVector-\completeMatrix\delayMtxSlack\amplitudeVectorSlack \rVert_2^2
	 \nonumber\\&=\arg\min_{
	\substack{ \{(\timeArrivalSlack[\targetindex], \reflectioncoefficientSlack[\targetindex])\} \\ \targetindex=1,\mydots,\numberoftargets}
	 } \lVert\completeMatrix\delayMtxSlack\amplitudeVectorSlack \rVert_2^2 -2 \Re\{\amplitudeVectorSlack^{\rm H}\delayMtxSlack^{\rm H}\completeMatrix^{\rm H}\rxSymbolsVector\}~.
	 \label{eq:MLest}
	\end{align}
\else
	\begin{align}
		\{(\timeArrivalEst[\targetindex],\reflectioncoefficientEst[\targetindex])\}&=\arg\min_{
			\substack{ \{(\timeArrivalSlack[\targetindex], \reflectioncoefficientSlack[\targetindex])\} \\ \targetindex=1,\mydots,\numberoftargets}
		} \lVert  \rxSymbolsVector-\completeMatrix\delayMtxSlack\amplitudeVectorSlack \rVert_2^2
		=\arg\min_{
			\substack{ \{(\timeArrivalSlack[\targetindex], \reflectioncoefficientSlack[\targetindex])\} \\ \targetindex=1,\mydots,\numberoftargets}
		} \lVert\completeMatrix\delayMtxSlack\amplitudeVectorSlack \rVert_2^2 -2 \Re\{\amplitudeVectorSlack^{\rm H}\delayMtxSlack^{\rm H}\completeMatrix^{\rm H}\rxSymbolsVector\}~.
		\label{eq:MLest}
	\end{align}
\fi
For a single target, i.e., $\numberoftargets=1$, \eqref{eq:MLest} can be reduced to
\begin{align}
\timeArrivalEst[1] = \arg\max_{\timeArrivalSlack[1]} |\Re \{\delayVector[{\timeArrivalSlack[1]}]^{\rm H} \completeMatrix^{\rm H}\rxSymbolsVector \}|~,
\label{eq:matchedFilter}
\end{align}
where $\reflectioncoefficientEst[1]=  \Re \{\delayVector[{\timeArrivalEst[1] }]^{\rm H} \completeMatrix^{\rm H}\rxSymbolsVector \}/(\symbolsVectorinFrequency^{\rm H}\symbolsVectorinFrequency)$ by equating the derivative of cost function with respect to $\timeArrivalSlack[1]$ and $ \reflectioncoefficientSlack[1]$ to zeros. The reason for the absolute value in \eqref{eq:matchedFilter} is that $\reflectioncoefficient[1]$ can be negative or positive.
The solution of \eqref{eq:matchedFilter} corresponds to the optimum \ac{MF} and the objective function can be evaluated via a computer search. Note that $\delayVector[{\timeArrivalEst[\targetindex] }]$ is a function of the carrier frequency. Thus, the search should consider narrow enough steps to obtain the maximum. In this study, we utilize a refinement procure that increases the number of points around the coarse estimate point through chirp Z-transformation.

The solution of \eqref{eq:MLest} is not trivial for $\numberoftargets>1$ and a brute-force search can cause a high-complexity \ac{RXr}. To address this issue, we utilize \eqref{eq:matchedFilter} and propose an iterative procedure by subtracting the information related to $(\indexIteration-1)$th  target from the signal as
\begin{align}
\rxSymbolsVector^{(\indexIteration)}=\rxSymbolsVector^{(\indexIteration-1)}-\reflectioncoefficientEst[\indexIteration-1]\completeMatrix\delayVector[{\timeArrivalEst[\indexIteration-1]}]~,
\label{eq:iterations}
\end{align}
for $\rxSymbolsVector^{(1)}=\rxSymbolsVector$. To increase accuracy further, after $\timeArrivalEst[\targetindex]$ and $\reflectioncoefficientEst[\targetindex]$ are estimated through iterations for $\targetindex=1,\mydots,\numberoftargets$,  we re-use the estimates obtained from \eqref{eq:iterations} and update $\timeArrivalEst[\indexIteration]$ and $\reflectioncoefficientEst[\indexIteration]$  by using
\begin{align}
	\rxSymbolsVector^{(\indexIteration)}=\rxSymbolsVector-\sum_{\substack{\targetindex=1\\\targetindex\neq\indexIteration}}^{\numberoftargets}\reflectioncoefficientEst[\targetindex]\completeMatrix\delayVector[{\timeArrivalEst[\targetindex]}]~,
	\label{eq:iterationsAfter}
\end{align}
in  \eqref{eq:matchedFilter}. The corresponding range for $\targetindex$th target can then be obtained as $\distanceEst[\targetindex]=\timeArrivalEst[\targetindex]\times\speedoflight/2$. Based on our simulation trials, updating the estimates through \eqref{eq:iterationsAfter} twice addresses the \ac{RMSE} floor in our previous results in \cite{Safi_2020_GC}.

\def\minimumResolution{r_{\text{min}}}
The successful cancellation of the ($\indexIteration$-1)th reflected signal in \eqref{eq:iterations} relies on the accurate estimate of the reflection coefficient. However, when there are multiple targets, the reflection coefficient estimation can be inaccurate due to 1) the distance between the targets and 2) the correlation properties of the waveform. The reason for the former issue is that multiple targets appear as a single target if the distance between two targets is less than the minimum resolution. It is well-known that the minimum resolution can be calculated as $\minimumResolution=0.5\times\speedoflight/\bandwidthchirp$ meters, where $\bandwidthchirp=\numberOfOccupiedSubcarriers/\symbolDuration$ is the chirp bandwidth. 
The issue related to the waveform can be seen in  \eqref{eq:matchedFilter}. For multiple \acp{CSC} transmission, the \ac{MF} output, i.e., $\Re \{\delayVector[{\timeArrivalEst[1] }]^{\rm H} \completeMatrix^{\rm H}\rxSymbolsVector \}$ for $\timeArrivalEst[1]\in[0,\CPDuration)$, is a superposition of the \ac{MF} outputs of all \acp{CSC} and \acp{CSC} that are closer to each other  in time can cause inaccurate estimations of the reflection coefficients. 
Also, the estimation accuracy for path delays can degrade since the reward function in \eqref{eq:matchedFilter} can be high at different values of $\timeArrivalSlack[1]$ for $\numberofIndices>1$, i.e., multiple spikes, although there is a single target, 
 To address the correlation problem, we investigate two solutions:  1) \ac{IS} unique to the proposed scheme  and 2) Range estimation over the \ac{LMMSE} channel estimate, i.e., removing the impact of waveform on the range estimation.

\def\distanceToBegin{a}
\def\distanceToEnd{b}
\def\upperLimit{U}
\def\cardinalityConstant{T}
\subsection{Solution \#1: Index Separation for \ac{MF}-based estimation}
\label{subsec:IS}
For this solution, we consider the \ac{MF}-based estimation and use the iterations in \eqref{eq:iterations} at the \ac{RXr}. However, we modify the transmitter such that the spikes in the \ac{AC} function occurring due to the multiple \ac{CSC} transmissions are well-separated. To this end, we restrict the selected \ac{CSC} indices as the distance between two adjacent indices is larger than a certain value, which separates \acp{CSC} apart  in time as in \figurename~\ref{fig:txrxa}.  This restriction improves the accuracy of the reflection coefficient estimation in \eqref{eq:matchedFilter} and the accuracy of the cancellations in \eqref{eq:iterations} by guaranteeing a low \ac{AC} zone.

\ac{IS} introduces a  trade-off between communications and radar. This is because a larger $\separationValue$ for improving the radar functionality can degrade the \ac{SE} of the proposed scheme. Hence, the first question that we need to address is how the \ac{SE} of \ac{CSC-IM} is affected for a given $\separationValue$. 
In addition, the restriction on indices  under \ac{IS} requires a new bijective mapping between bits and indices  for the encoder and  decoder designs.

\subsubsection{Spectral Efficiency under IS}
Let $\cardinality[\numberofIndices,\separationValue][\numberOfShifts]$ denote the cardinality of the set of sequences $(\selectedChirpIndex[0],\selectedChirpIndex[1],\mydots,\selectedChirpIndex[\numberofIndices-1])$, where  $0\le\selectedChirpIndex[\chirpm]<\selectedChirpIndex[\chirpn]<\numberOfShifts$ for $\chirpm<\chirpn$ and
 $\separationDis[\indexSeparation]\ge\separationValue$ for all $\indexSeparation\in\{1,2,\mydots,\numberofIndices\}$, i.e., the number of valid index permutations for given $\separationValue$, $\numberofIndices$ and $\numberOfShifts$. Also, let $\CardinalityDis[\numberofIndices,\separationValue][\distanceSumVar]$ denote the cardinality of the sequences $(\separationDis[1],\separationDis[2],\cdots,\separationDis[\numberofIndices])$ such that $\separationDis[1]+\separationDis[2]+\cdots+\separationDis[\numberofIndices]=\distanceSumVar\in\integersNonnegativeSet$ and  $\separationDis[\indexSeparation]\ge\separationValue\in\integersNonnegativeSet$ for $\indexSeparation\in\{1,2,\mydots,\numberofIndices\}$. To obtain $\cardinality[\numberofIndices,\separationValue][\numberOfShifts]$, we need the following lemma:  
\begin{lemma}
	\label{th:cardinalitySum}
	For $\numberofIndices\ge1$,
	\begin{align}
		\CardinalityDis[\numberofIndices,\separationValue][\distanceSumVar]=&
		\sum_{\indexRecursion=\separationValue}^{\distanceSumVar-\separationValue\numberofIndices+\separationValue} \CardinalityDis[\numberofIndices-1,\separationValue][\distanceSumVar-{\indexRecursion}]
		\label{eq:recursiveSumB}\\
		=&\begin{cases}
			\binom{\distanceSumVar-\numberofIndices\separationValue+\numberofIndices-1}{\numberofIndices-1},&~\distanceSumVar\ge\separationValue\numberofIndices\\
			0. & \text{\rm otherwise}
		\end{cases}~.
		\label{eq:cardinaltySum}
	\end{align}
\end{lemma}
\begin{proof}
	The cardinality of the set of  $(\separationDis[1],\separationDis[2],\cdots,\separationDis[\numberofIndices-1])$ is $\CardinalityDis[\numberofIndices-1,\separationValue][\distanceSumVar-{\separationDis[\numberofIndices]}]$ as $\separationDis[1]+\separationDis[2]+\cdots+\separationDis[\numberofIndices-1]=\distanceSumVar-\separationDis[\numberofIndices]$ for $\separationDis[\numberofIndices]\in\{\separationValue,\mydots,\distanceSumVar-\separationValue\numberofIndices+\separationValue\}$, which implies the recursive formula in \eqref{eq:recursiveSumB}.
	
	It is well-known that the number of compositions of $n$ into exactly $k$ parts is $\binom{n-1}{k-1}$, where each part is greater than $0$.We define the variable $\separationDis[\indexSeparation]'$ by setting  $\separationDis[\indexSeparation]'=\separationDis[\indexSeparation]-(\separationValue-1)$. 
	\begin{itemize}
			 \item Case 1 ($\separationValue\ge1$):  $\separationDis[1]'+\separationDis[2]'+\cdots+\separationDis[\numberofIndices]'=\distanceSumVar-\numberofIndices(\separationValue-1)$ holds. Hence, $\CardinalityDis[\numberofIndices,\separationValue][\distanceSumVar]$ must be equal to the number of compositions of $\distanceSumVar-\numberofIndices(\separationValue-1)$ into exactly $\numberofIndices$ parts, which implies \eqref{eq:cardinaltySum} for $\separationValue>1$. 
			 \item Case 2 ($\separationValue=0$): Since $\separationDis[1]'+\separationDis[2]'+\cdots+\separationDis[\numberofIndices]'=\distanceSumVar+\numberofIndices$ holds, $\CardinalityDis[\numberofIndices,1][\distanceSumVar]$ must be equal to $\binom{\distanceSumVar+\numberofIndices-1}{\numberofIndices-1}$.
	\end{itemize}
	If $\distanceSumVar<\separationValue\numberofIndices$, there exists no composition.
\end{proof}

\begin{theorem}
	For $\numberofIndices\ge2$,
	\begin{align}
		\cardinality[\numberofIndices,\separationValue][\numberOfShifts]=\begin{cases}
		\frac{\numberOfShifts}{\numberofIndices}\binom{\numberOfShifts-\numberofIndices\separationValue-1}{\numberofIndices-1},&\numberOfShifts\ge\numberofIndices(\separationValue+1)\\
		0,& \text{\rm otherwise}
		\end{cases}~.
		\label{eq:cardinalityA}
	\end{align}
		\label{th:safisEquation}
\end{theorem}

\begin{proof}
Consider the following steps:
	
 Step 1: By the definition in \eqref{eq:sepDef}, the cardinality of the set of sequences $(\selectedChirpIndex[0],\separationDis[1],\mydots,\separationDis[\numberofIndices-1])$ is $\cardinality[\numberofIndices,\separationValue][\numberOfShifts]$ and
		$\separationDis[1]+\separationDis[2]+\cdots+\separationDis[\numberofIndices]=\numberOfShifts-\numberofIndices$. Since $\separationDis[\indexSeparation]\ge\separationValue$ for all $\indexSeparation$, the inequality given by
		\begin{align}
		\separationValue\le\separationDis[\numberofIndices]\le\upperLimit~,
		\label{eq:boundss}
		\end{align}
		holds for $\upperLimit\triangleq\numberOfShifts-\numberofIndices(\separationValue+1)+\separationValue$. 
		
		Step 2: By the definition in \eqref{eq:sepDef}, $\separationDis[\numberofIndices]$ can be expressed as $\separationDis[\numberofIndices]=\selectedChirpIndex[0]+\distanceToEnd$, where 
		$\distanceToEnd=\numberOfShifts-1-\selectedChirpIndex[\numberofIndices-1]$. 
		Therefore, for $\selectedChirpIndex[0]\in\{0,1,\mydots,\upperLimit\}$, the inequality in \eqref{eq:boundss} can be re-stated as
		\begin{align}
		\separationValue-\selectedChirpIndex[0]\le\distanceToEnd \le \upperLimit-\selectedChirpIndex[0]~.
		\label{eq:boundsss}
		\end{align}
		%
		
 Step 3: Since $\separationDis[\indexSeparation]\ge\separationValue$ for all $\indexSeparation$ and 		$\separationDis[1]+\separationDis[2]+\cdots+\separationDis[\numberofIndices]=\numberOfShifts-\numberofIndices$, the cardinality of the sequences $(\separationDis[1],\separationDis[2],\cdots,\separationDis[\numberofIndices-1])$ is $\CardinalityDis[\numberofIndices-1,\separationValue][\numberOfShifts-\numberofIndices-{\separationDis[\numberofIndices]}]$ for a given $\separationDis[\numberofIndices]$. 
		Hence, by using \eqref{eq:boundsss} and  Lemma~\ref{th:safisEquation}, $\cardinality[\numberofIndices,\separationValue][\numberOfShifts]$ can be expressed as
		\if \IEEEsubmission 0
			\begin{align}
			\cardinality[\numberofIndices,\separationValue][\numberOfShifts]&=\sum_{\selectedChirpIndex[0]=0}^{\upperLimit}\sum_{\distanceToEnd=\max(0,\separationValue-\selectedChirpIndex[0])}^{\upperLimit-\selectedChirpIndex[0]}\CardinalityDis[\numberofIndices-1,\separationValue][\numberOfShifts-\numberofIndices-{(\selectedChirpIndex[0]+\distanceToEnd)}]~,\nonumber\\
			=&\sum_{{\selectedChirpIndex[0]}=0}^{\separationValue-1}\CardinalityDis[\numberofIndices,\separationValue][\numberOfShifts-\numberofIndices]+\sum_{{\selectedChirpIndex[0]}=\separationValue}^{\upperLimit}\CardinalityDis[\numberofIndices,\separationValue][\numberOfShifts-\numberofIndices+\separationValue-{\selectedChirpIndex[0]}]\label{eq:shortcut}\\
			=&\separationValue\CardinalityDis[\numberofIndices,\separationValue][\numberOfShifts-\numberofIndices]+\CardinalityDis[\numberofIndices+1,\separationValue][\numberOfShifts-\numberofIndices+\separationValue]\nonumber\\
			=&\frac{\numberOfShifts}{\numberofIndices}\binom{\numberOfShifts-\numberofIndices\separationValue-1}{\numberofIndices-1}~.\nonumber
			\end{align}
		\else
			\begin{align}
			\cardinality[\numberofIndices,\separationValue][\numberOfShifts]&=\sum_{\selectedChirpIndex[0]=0}^{\upperLimit}\sum_{\distanceToEnd=\max(0,\separationValue-\selectedChirpIndex[0])}^{\upperLimit-\selectedChirpIndex[0]}\CardinalityDis[\numberofIndices-1,\separationValue][\numberOfShifts-\numberofIndices-{(\selectedChirpIndex[0]+\distanceToEnd)}]~,\nonumber\\
			=&\sum_{{\selectedChirpIndex[0]}=0}^{\separationValue-1}\CardinalityDis[\numberofIndices,\separationValue][\numberOfShifts-\numberofIndices]+\sum_{{\selectedChirpIndex[0]}=\separationValue}^{\upperLimit}\CardinalityDis[\numberofIndices,\separationValue][\numberOfShifts-\numberofIndices+\separationValue-{\selectedChirpIndex[0]}]\label{eq:shortcut}\\
			=&\separationValue\CardinalityDis[\numberofIndices,\separationValue][\numberOfShifts-\numberofIndices]+\CardinalityDis[\numberofIndices+1,\separationValue][\numberOfShifts-\numberofIndices+\separationValue]
			=\frac{\numberOfShifts}{\numberofIndices}\binom{\numberOfShifts-\numberofIndices\separationValue-1}{\numberofIndices-1}~.\nonumber
			\end{align}
		\fi
 There exists no valid sequence for $\numberOfShifts<\numberofIndices(\separationValue+1)$, i.e., $\cardinality[\numberofIndices,\separationValue][\numberOfShifts]=0$. 
\end{proof}

For a given $\separationValue$, the \ac{SE} of the \ac{CSC-IM} can now be calculated as $\spectralEfficiency=\floor{{\rm log}_2(\cardinality[\numberofIndices,\separationValue][\numberOfShifts]\times \numberOfPointsForPSK^\numberofIndices)}/\numberOfShifts$. In \figurename~\ref{fig:IS}\subref{subfig:numberOfBits}, we show the trade-off between $\separationValue$ and the maximum number of information bits that are encoded with the indices, i.e., ${\rm log}_2(\cardinality[\numberofIndices,\separationValue][\numberOfShifts])$. As expected, a larger $\separationValue$ means a lower number of information bits that can be transmitted on the indices. The degradation is more rapid with a larger $\numberofIndices$ although the number of information bits is larger for smaller values of $\separationValue$.

\begin{figure*}
	\centering
	\subfloat[Trade-off between separation and the number of information bits on the indices.]{\includegraphics[width =3.3in]{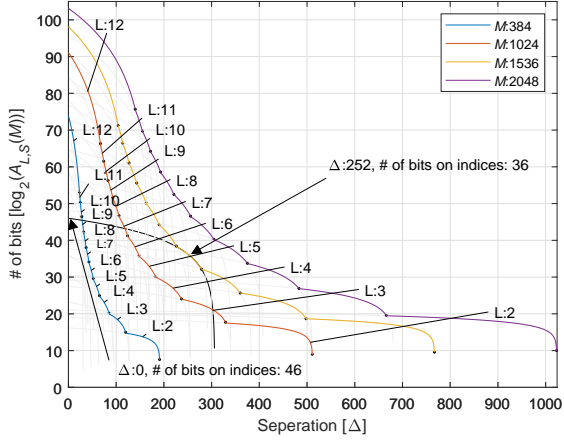}
		\label{subfig:numberOfBits}}~
	\subfloat[$\separationValueMax$ versus $\numberOfShifts$. The distance between two indices can be as large as $\numberOfShifts/4$ without losing \ac{SE} for $\numberofIndices=2$.]{\includegraphics[width =3.3in]{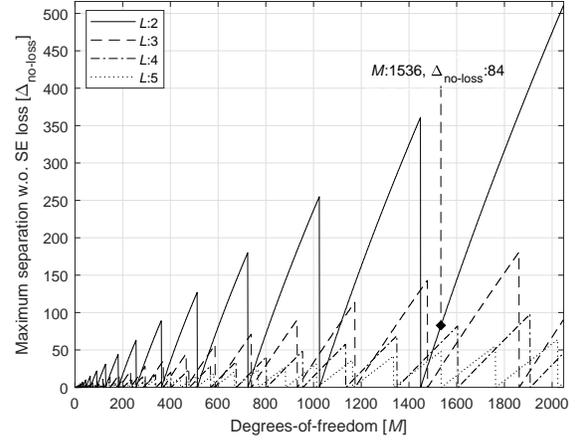}
	\label{subfig:smaxvsm}}~
	\caption{The impact of \ac{IS} on the \ac{SE}.}
	\label{fig:IS}
\if\IEEEsubmission1
\vspace{-3mm}
\fi
\end{figure*}
One interesting question is that what is the largest $\separationValue$ such that the \ac{SE} still remains at the maximum for a given $\numberofIndices$? To address this question, let	$\separationValueMax$ denote the largest separation such that $\floor{\log_{2}{\binom{\numberOfShifts}{\numberofIndices}}}=\floor{\log_{2}{\cardinality[\numberofIndices,\separationValueMax][\numberOfShifts]}}$. Hence,  $\cardinality[\numberofIndices,\separationValueMax][\numberOfShifts]\geq2^{\floor{\log_{2}{\binom{\numberOfShifts}{\numberofIndices}}}}$ must hold. Since Theorem~\ref{th:safisEquation} provides a closed-form solution, $\separationValueMax$ can be evaluated with a computer search.

In \figurename~\ref{fig:IS}\subref{subfig:smaxvsm}, we plot $\separationValueMax$ for a given $\numberOfShifts$. The surprising result is that the separation between two adjacent indices can be as large as $\numberOfShifts/4-1$ without losing \ac{SE} for $\numberofIndices=2$. For instance, for $\numberOfShifts=2^{\naturalnum}$, where $\naturalnum\in\integersPositiveSet$, $\separationValueMax$ reaches its maximum value, i.e., $\separationValueMax=\numberOfShifts/4-1$. In other words, the duration of the low \ac{AC}  zone can be as large as a typical \ac{CP} size  as $\separationValueMax/ \numberOfShifts = \CPSize/\idftSize$ can be maintained.
The value of $\numberOfShifts$ where $\separationValueMax$ reaches at its maximum  depends on $\numberofIndices$. For instance, the values of $\numberOfShifts$ are $931$, $954$, and $1012$ for $\numberofIndices=3$,  $\numberofIndices=4$, and $\numberofIndices=5$, respectively, and the corresponding values for $\separationValueMax$ are $90$, $48$, and $31$.
The ratios between $\separationValueMax$ and $\numberOfShifts$ approach to $1/10.25$, $1/19.4$, and $1/32.45$ for $\numberofIndices=3$, $\numberofIndices=4$, and $\numberofIndices=5$, respectively. 
We also observe abrupt changes in $\separationValueMax$ for different values of $\numberOfShifts$. For example, $\separationValueMax$ becomes minimum, i.e., $\separationValueMax=0$, for $\numberOfShifts=2^{\naturalnum}+1$ for $\numberOfShifts=2$. This behavior is because the number of bits encoded on the  indices increases by $1$ when  $\numberOfShifts$ increases by $1$.

\subsubsection{Bijective mappings between natural numbers and indices under IS}
\label{subsubsec:bijection}
To develop an encoder and a decoder for \ac{CSC-IM} by taking the \ac{IS} into account, information bits need to be mapped to the indices under a separation constraint or vice versa. We address this issue by deriving the mapping rules from a positive integer $\anaturalNumberToSeqeunceA$ to the indices  or vice versa for a given $\separationValue$, where the information bits can be converted to $\anaturalNumberToSeqeunceA$ through a binary to decimal conversion. The following definitions are needed:
\begin{definition} \rm  The  function  
 $\algorithmEncoderA[\anaturalNumberToSeqeunceA,\numberOfShifts,\numberofIndices,\separationValue]$ maps a positive integer $\anaturalNumberToSeqeunceA$ to the sequence $(\selectedChirpIndex[0],\selectedChirpIndex[1],\mydots,\selectedChirpIndex[\numberofIndices-1])$ for given $\numberOfShifts$, $\numberofIndices$, and $\separationValue$.
\end{definition}

\begin{definition}
	\rm The function
	 $\algorithmEncoderB[\anaturalNumberToSeqeunceB,\distanceSumVar,\numberofIndices,\separationValue]$  maps a positive integer $\anaturalNumberToSeqeunceB$  to the sequence  $(\separationDis[1],\separationDis[2],\mydots,\separationDis[\numberofIndices])$ for given $\distanceSumVar$, $\numberofIndices$, and $\separationValue$.
\end{definition}

\begin{definition} \rm
The inverse functions of $\algorithmEncoderA[\anaturalNumberToSeqeunceA,\numberOfShifts,\numberofIndices,\separationValue]$ and $\algorithmEncoderB[\anaturalNumberToSeqeunceB,\distanceSumVar,\numberofIndices,\separationValue]$ are $\algorithmDecoderA[{(\selectedChirpIndex[0],\selectedChirpIndex[1],\mydots,\selectedChirpIndex[\numberofIndices-1])},\numberOfShifts,\numberofIndices,\separationValue]$ and $\algorithmDecoderB[{(\separationDis[1],\mydots,\separationDis[\numberofIndices])},\distanceSumVar,\numberofIndices,\separationValue]$, respectively.
\end{definition}

\if\IEEEsubmission0
\def\baselineSize{1}
\renewcommand{\baselinestretch}{1}
\else
\def\baselineSize{1.5}
\renewcommand{\baselinestretch}{0.8}
\fi
\begin{algorithm}[t]{
		\scriptsize
		\caption{\small Mapping algorithms for $\separationValue\ge0$}\label{alg:enum}
		\SetKwInput{KwInput}{Input}                
		\SetKwInput{KwOutput}{Output}              
		\DontPrintSemicolon
		
		
		\SetKwFunction{FencA}{$\algorithmEncoderA[\anaturalNumberToSeqeunceA,\numberOfShifts,\numberofIndices,\separationValue]$}
		\SetKwFunction{FencB}{$\algorithmEncoderB[\anaturalNumberToSeqeunceB,\distanceSumVar,\numberofIndices,\separationValue]$}	
		\SetKwFunction{FdecA}{$\algorithmDecoderA[{(\selectedChirpIndex[0],\selectedChirpIndex[1],\mydots,\selectedChirpIndex[\numberofIndices-1])},\numberOfShifts,\numberofIndices,\separationValue]$}
		\SetKwFunction{FdecB}{$\algorithmDecoderB[{(\separationDis[1],\mydots,\separationDis[\numberofIndices])},\distanceSumVar,\numberofIndices,\separationValue]$}

		\SetKwProg{Fn}{Function}{}{}
		\Fn{${(\selectedChirpIndex[0],\selectedChirpIndex[1],\mydots,\selectedChirpIndex[\numberofIndices-1])}=\FencA$}{
		Obtain the largest $\anumber$ for $\anumberAfterSum[\anumber]=\sum_{{\distanceToBegin}=0}^{\anumber}\functionB[{\distanceToBegin}]<\anaturalNumberToSeqeunceA$\;
		$\selectedChirpIndex[0]=\anumber$\;
		\eIf{${\selectedChirpIndex[0]}<\separationValue$}{
			 $(\separationDis[1],\mydots,\separationDis[\numberofIndices])=\algorithmEncoderB[{\anumberAfterSum[{\selectedChirpIndex[0]}]}-\anaturalNumberToSeqeunceA,\numberOfShifts-\numberofIndices,\numberofIndices,\separationValue]$
		}{
			 $(\separationDis[1],\mydots,\separationDis[\numberofIndices])=\algorithmEncoderB[{\anumberAfterSum[{\selectedChirpIndex[0]}]}-\anaturalNumberToSeqeunceA, {\numberOfShifts-\numberofIndices+\separationValue-\selectedChirpIndex[0]},
			\numberofIndices,\separationValue]$ 
		}
		Obtain $(\selectedChirpIndex[0],\selectedChirpIndex[1],\mydots,\selectedChirpIndex[\numberofIndices-1])$ from $(\selectedChirpIndex[0],\separationDis[1],\mydots,\separationDis[\numberofIndices-1])$\;
		} 
		\SetKwProg{Fn}{Function}{}{}
		\Fn{${(\separationDis[1],\separationDis[2],\mydots,\separationDis[\numberofIndices])}=\FencB$}{
		\eIf{$\numberofIndices>1$}
		{
			Obtain the largest $\anumber$ for $\anumberAfterSumB[\anumber]=\sum_{\indexRecursion=\separationValue}^{\anumber} \CardinalityDis[\numberofIndices-1,\separationValue][\distanceSumVar-{\indexRecursion}]<\anaturalNumberToSeqeunceB$\;
			$\separationDis[\numberofIndices]=\anumber$\;			
			 $(\separationDis[1],\mydots,\separationDis[\numberofIndices-1])=\algorithmEncoderB[\anaturalNumberToSeqeunceB-{\anumberAfterSumB[\anumber]},\distanceSumVar-{\separationDis[\numberofIndices]},\numberofIndices-1,\separationValue]$
		}{
			$\separationDis[1]=\anaturalNumberToSeqeunceB$\;		 
		}
}
		\SetKwProg{Fn}{Function}{}{}
\Fn{$\anaturalNumberToSeqeunceA=\FdecA$}{
	Obtain $(\separationDis[1],\mydots,\separationDis[\numberofIndices])$ from $(\selectedChirpIndex[0],\selectedChirpIndex[1],\mydots,\selectedChirpIndex[\numberofIndices-1])$\;
	\eIf{${\selectedChirpIndex[0]}<\separationValue$}
	{
 $\anaturalNumberToSeqeunceB=\algorithmDecoderB[{(\separationDis[1],\mydots,\separationDis[\numberofIndices])},\numberOfShifts-\numberofIndices,\numberofIndices,\separationValue]$
	}{
				 $\anaturalNumberToSeqeunceB=\algorithmDecoderB[{(\separationDis[1],\mydots,\separationDis[\numberofIndices])},{\numberOfShifts-\numberofIndices+\separationValue-\selectedChirpIndex[0]},\numberofIndices,\separationValue]$
	}
	$\anaturalNumberToSeqeunceA=\anumberAfterSum[{\selectedChirpIndex[0]-1}]+\anaturalNumberToSeqeunceB$
}
		\SetKwProg{Fn}{Function}{}{}
\Fn{$\anaturalNumberToSeqeunceB=\FdecB$}{

	\eIf{$\numberofIndices>1$}
	{
	$\anaturalNumberToSeqeunceB=\anumberAfterSumB[{\separationDis[\numberofIndices]}]+\algorithmDecoderB[{(\separationDis[1],\mydots,\separationDis[\numberofIndices-1])},\distanceSumVar-{\separationDis[\numberofIndices]},\numberofIndices-1,\separationValue]$
	}{
	$\anaturalNumberToSeqeunceB= 1$
	}
}
	}
\end{algorithm}
\begin{table}
	\caption{The enumerations of $\algorithmEncoderA[\anaturalNumberToSeqeunceA,\numberOfShifts,\numberofIndices,\separationValue]$ for $\numberOfShifts=10$, $\numberofIndices=3$, and $\separationValue\in\{0,1,2\}$.}
	\centering
	\begin{tabular}{c||ccc||ccc||ccc}
		& \multicolumn{3}{c||}{$\separationValue=0$}                                     & \multicolumn{3}{c||}{$\separationValue=1$}                                     & \multicolumn{3}{c}{$\separationValue=2$}                                        \\ 
		\hhline{=::===::===::===}
		$n$                   & $\selectedChirpIndex[0]$ & $\selectedChirpIndex[1]$ & $\selectedChirpIndex[2]$ & $\selectedChirpIndex[0]$ & $\selectedChirpIndex[1]$ & $\selectedChirpIndex[2]$ & $\selectedChirpIndex[0]$ & $\selectedChirpIndex[1]$ & $\selectedChirpIndex[2]$  \\ 
		\hline
		1                     & 0                        & 8                        & 9                        & 0                        & 6                        & 8                        & 0                        & 4                        & 7                         \\
		2                     & 0                        & 7                        & 9                        & 0                        & 5                        & 8                        & 0                        & 3                        & 7                         \\
		3                     & 0                        & 6                        & 9                        & 0                        & 4                        & 8                        & 0                        & 3                        & 6                         \\
		4                     & 0                        & 5                        & 9                        & 0                        & 3                        & 8                        & 1                        & 5                        & 8                         \\
		5                     & 0                        & 4                        & 9                        & 0                        & 2                        & 8                        & 1                        & 4                        & 8                         \\
		6                     & 0                        & 3                        & 9                        & 0                        & 5                        & 7                        & 1                        & 4                        & 7                         \\
		7                     & 0                        & 2                        & 9                        & 0                        & 4                        & 7                        & 2                        & 6                        & 9                         \\
		8                     & 0                        & 1                        & 9                        & 0                        & 3                        & 7                        & 2                        & 5                        & 9                         \\
		9                     & 0                        & 7                        & 8                        & 0                        & 2                        & 7                        & 2                        & 5                        & 8                         \\
		10                    & 0                        & 6                        & 8                        & 0                        & 4                        & 6                        & 3                        & 6                        & 9                         \\ 
		\cline{8-10}
		$\vdots$ & \multicolumn{3}{c||}{$\vdots$}                                                 & \multicolumn{3}{c||}{$\vdots$}                                                 & \multicolumn{3}{c}{}                                                            \\
		50                    & 1                        & 6                        & 7                        & 5                        & 7                        & 9                        & \multicolumn{3}{c}{N/A}                                                         \\ 
		\cline{5-7}
		$\vdots$              & \multicolumn{3}{c||}{$\vdots$}                                                 & \multicolumn{3}{c||}{N/A}                                                      & \multicolumn{3}{c}{}                                                            \\
		120                   & 7                        & 8                        & 9                        & \multicolumn{3}{c||}{}                                                         & \multicolumn{3}{c}{}                                                            \\ 
		\hline
		Cardinality           & \multicolumn{3}{c||}{120}                                                      & \multicolumn{3}{c||}{50}                                                       & \multicolumn{3}{c}{10}                                                         
	\end{tabular}
	\label{table:example}
\end{table}
\renewcommand{\baselinestretch}{\baselineSize}

For $\algorithmEncoderA[\anaturalNumberToSeqeunceA,\numberOfShifts,\numberofIndices,\separationValue]$, we use the expansion in \eqref{eq:shortcut}.
We first determine $\selectedChirpIndex[0]$ as the maximum value of $\anumber$ such that $\anumberAfterSum[\anumber]\triangleq\sum_{{\selectedChirpIndex[0]}=0}^{\anumber}\functionB[{\selectedChirpIndex[0]}]<\anaturalNumberToSeqeunceA$, where $\functionB[{\selectedChirpIndex[0]}]\triangleq\CardinalityDis[\numberofIndices,\separationValue][\numberOfShifts-\numberofIndices+\min(0,\separationValue-{\selectedChirpIndex[0]})]$. We then obtain $(\separationDis[1],\mydots,\separationDis[\numberofIndices])$ by using $\algorithmEncoderB[{\anumberAfterSum[{\selectedChirpIndex[0]}]}-\anaturalNumberToSeqeunceA,\numberOfShifts-\numberofIndices,\numberofIndices,\separationValue]$ for  ${\selectedChirpIndex[0]}<\separationValue$ and 
 $\algorithmEncoderB[{\anumberAfterSum[{\selectedChirpIndex[0]}]}-\anaturalNumberToSeqeunceA, {\numberOfShifts-\numberofIndices+\separationValue-\selectedChirpIndex[0]},
 \numberofIndices,\separationValue]$  for  ${\selectedChirpIndex[0]}\ge\separationValue$. We finally calculate  $\selectedChirpIndex[\indexChirp]=\selectedChirpIndex[0]+\sum_{j=1}^{\indexChirp}(1+\separationDis[j])$ for $\indexChirp\in\{1,2,\mydots,\numberofIndices-1\}$. 
 
 For $\algorithmEncoderB[\anaturalNumberToSeqeunceB,\distanceSumVar,\numberofIndices,\separationValue]$, we exploit the sum in \eqref{eq:recursiveSumB} and obtain $\separationDis[\numberofIndices]$ as the maximum value of $\anumber$ such that $\anumberAfterSumB[\anumber]\triangleq\sum_{\indexRecursion=\separationValue}^{\anumber} \CardinalityDis[\numberofIndices-1,\separationValue][\distanceSumVar-{\indexRecursion}]<\anaturalNumberToSeqeunceB$. Since determining  $\separationDis[\numberofIndices]$ reduces the original problem from $\numberofIndices$ parts to $\numberofIndices-1$ parts, where the new sum is $\distanceSumVar-\separationDis[\numberofIndices]$, $\algorithmEncoderB[\anaturalNumberToSeqeunceB,\distanceSumVar,\numberofIndices,\separationValue]$ recalls itself as $\algorithmEncoderB[\anaturalNumberToSeqeunceB-{\anumberAfterSumB[\anumber]},\distanceSumVar-{\separationDis[\numberofIndices]},\numberofIndices-1,\separationValue]$ to obtain $(\separationDis[1],\mydots,\separationDis[\numberofIndices-1])$. This procedure is recursive and continues till $\numberofIndices=1$. For $\numberofIndices=1$, $\algorithmEncoderB[\anaturalNumberToSeqeunceB,\distanceSumVar,\numberofIndices,\separationValue]$ returns $\separationDis[1]=\anaturalNumberToSeqeunceB$.
 
The function $\algorithmDecoderA[{(\selectedChirpIndex[0],\selectedChirpIndex[1],\mydots,\selectedChirpIndex[\numberofIndices-1])},\numberOfShifts,\numberofIndices,\separationValue]$ first calculates $\anumberAfterSum[{\selectedChirpIndex[0]-1}]$. Afterwards, it obtains $(\separationDis[1],\mydots,\separationDis[\numberofIndices-1])$ from $(\selectedChirpIndex[0],\selectedChirpIndex[1],\mydots,\selectedChirpIndex[\numberofIndices-1])$. 
 Finally, it returns $\anaturalNumberToSeqeunceA=\anumberAfterSum[{\selectedChirpIndex[0]}]+\anaturalNumberToSeqeunceB$, where $\anaturalNumberToSeqeunceB$ is $\algorithmDecoderB[{(\separationDis[1],\mydots,\separationDis[\numberofIndices])},\numberOfShifts-\numberofIndices,\numberofIndices,\separationValue]$ for $\selectedChirpIndex[0]<\separationValue$ and $\algorithmDecoderB[{(\separationDis[1],\mydots,\separationDis[\numberofIndices])},{\numberOfShifts-\numberofIndices+\separationValue-\selectedChirpIndex[0]},\numberofIndices,\separationValue]$ for $\selectedChirpIndex[0]\ge\separationValue$ based on \eqref{eq:shortcut}. 

The function $\algorithmDecoderB[{(\separationDis[1],\mydots,\separationDis[\numberofIndices])},\distanceSumVar,\numberofIndices,\separationValue]$ first calculates $\anumberAfterSumB[{\separationDis[\numberofIndices]}]$. It then returns the result as  $\anaturalNumberToSeqeunceB=\anumberAfterSumB[{\separationDis[\numberofIndices]}]+\algorithmDecoderB[{(\separationDis[1],\mydots,\separationDis[\numberofIndices-1])},\numberOfShifts-{\separationDis[\numberofIndices]},\numberofIndices-1,\separationValue]$. For $\numberofIndices=1$, $\algorithmDecoderB[{(\separationDis[1],\mydots,\separationDis[\numberofIndices])},\distanceSumVar,\numberofIndices,\separationValue]$ is $1$.

The pseudocodes for the mapping algorithms are provided in Algorithm~\ref{alg:enum}. As the closed-form expressions of $\cardinality[\numberofIndices,\separationValue][\numberOfShifts]$ and $\CardinalityDis[\numberofIndices,\separationValue][\distanceSumVar]$ are available in Theorem~\ref{th:safisEquation} and Lemma~\ref{th:cardinalitySum}, respectively, the time complexity of these algorithms linearly scales with $\numberOfShifts$, $\numberofIndices$, and $1/\separationValue$. The algorithms can also be efficiently implemented as they rely on recursions. 

In \tablename~\ref{table:example}, we exemplify the output of $\algorithmEncoderA[\anaturalNumberToSeqeunceA,\numberOfShifts,\numberofIndices,\separationValue]$ for $\numberOfShifts=10$, $\numberofIndices=3$, and $\separationValue\in\{0,1,2\}$. For instance, for $\separationValue=2$, there are at least $2$ integers between any two adjacent indices and there are $10$ valid sequences. Hence, $3$ information bits can be encoded by using the decimal number converted from the binary number constructed with the information bits.

\subsubsection{Impact of IS on the communication receiver performance}
\label{subsubsec:ISerrorrate}
The \ac{ML} detector under the \ac{IS} can be expressed as
\begin{align}
    \{ \indexSequenceDetect, \pskSerialSequenceDetect\} = \arg\max_{\substack{\selectedChirpIndexDomain[\indexChirp]\in\{0,\mydots,\numberOfShifts-1\}\\ 		 	    
    		\symbolPSKserialdomain[\indexChirp]\in\integers_\numberOfPointsForPSK  \\ 
    		\selectedChirpIndexDomain[\chirpm]<\selectedChirpIndexDomain[\chirpn]  \textrm{ for }  \chirpm<\chirpn \\
    		{\separationDis[\indexSeparation]\ge\separationValue  \textrm{ for } \indexSeparation\in{1,\mydots,\numberofIndices}} }} \Re\left\{ \sum_{\indexChirp=0}^{\numberofIndices -1}{\dataSymbolAfterIDFTspread[{\selectedChirpIndexDomain[\indexChirp]}]\constante^{-\constantj2\pi\symbolPSKserialdomain[\indexChirp]/\numberOfPointsForPSK}}\right\}~,
    \label{eq:mldetectorIS}
\end{align}
where the condition $\separationDis[\indexSeparation]\ge\separationValue$ reduces the search space. A low-complexity receiver based on \eqref{eq:mldetectorIS} can be implemented as follows:
\begin{itemize}
\item Obtain $\{\selectedChirpIndexDetect[0],\symbolPSKserialdetect[0]\}$ that maximizes $\datasymbolestimate[\indexDataSymbolRX][\indexPSK]$ for $\indexDataSymbolRX\in\{0,1,\mydots,\numberOfShifts-1\}$  and $\indexPSK\in\integers_\numberOfPointsForPSK$.
\item Calculate $\{\selectedChirpIndexDetect[\indexChirp],\symbolPSKserialdetect[\indexChirp]\}$ that maximizes $\datasymbolestimate[\indexDataSymbolRX][\indexPSK]$  for $\indexDataSymbolRX\in\{0,1,\mydots,\numberOfShifts-1\}$  and $\indexPSK\in\integers_\numberOfPointsForPSK$ such that   $\min(|\selectedChirpIndexDetect[\indexChirp]-\selectedChirpIndexDetect[\indexChirp']|,\numberOfShifts-|\selectedChirpIndexDetect[\indexChirp]-\selectedChirpIndexDetect[\indexChirp']|)\ge\separationValue+1$ for all $0\le\indexChirp'\le\indexChirp-1$ till detecting the $(\numberofIndices-1)$th index and the corresponding \ac{PSK} symbol. 
\item Re-order the detected chirp and \ac{PSK} symbol indices such that $\selectedChirpIndexDetect[\chirpm]<\selectedChirpIndexDetect[\chirpn]  \textrm{ for }  \chirpm<\chirpn$.
\end{itemize}
Note that the \ac{IS} can slightly decrease the error rate since it restricts the valid index combinations.

\subsection{Solution \#2: Range Estimation over LMMSE Channel Estimate }
For this solution, we remove the impact of the waveform on the range estimation by using the \ac{LMMSE} estimate of $\channelFVector$, i.e., $\channelFVectorEst=\completeMatrix^{\rm H}(\completeMatrix\completeMatrix^{\rm H}+\noiseVariance \identityMatrix)^{-1}\rxSymbolsVector$ . For a single target, the estimate of $\timeArrivalEst[1]$  can then be obtained as
\begin{align}
	\timeArrivalEst[1] = \arg\max_{\timeArrivalSlack[1]} |\Re \{\delayVector[{\timeArrivalSlack[1]}]^{\rm H} \completeMatrix^{\rm H}(\completeMatrix\completeMatrix^{\rm H}+\noiseVariance \identityMatrix)^{-1}\rxSymbolsVector \}|~,
\end{align}
where $\reflectioncoefficientEst[1]=  \Re \{\delayVector[{\timeArrivalEst[1] }]^{\rm H} \completeMatrix^{\rm H}\rxSymbolsVector \}/(\symbolsVectorinFrequency^{\rm H}\symbolsVectorinFrequency+\noiseVariance)$.  For multiple targets, we also consider the iterative procedure in \eqref{eq:iterations} and \eqref{eq:iterationsAfter}. The main disadvantage of this method is that it does not attain the \ac{CRLB} of ranges if the waveform in the frequency domain is not unimodular as demonstrated  in Section~\ref{sec:numresults}. In addition, it requires an accurate estimation of the noise variance.

\def\realPartU[#1][#2]{\mu_{#1#2}}
\def\imagPartV[#1][#2]{\nu_{#1#2}}
\def\fisherInformationEle[#1][#2]{J_{#1#2}}
\def\unknownVector{{\rm \bf p}}
\def\unknownVectorEle[#1]{p_{#1}}
\def\RMSE{\sigma_{\text{range}}}
\def\RMSEref{\sigma_{\text{coeff}}}
\def\expectedValue[#1]{\mathbb{E}\left\{#1\right\}}
\subsection{CRLB for Range and Reflection Coefficients}
\label{subsec:CRLB}
To derive \ac{CRLB} for range and reflection coefficients, we follow a similar approach proposed in \cite{Rife_1974}. We first re-express $\receivedSignalDiscreteFrequency[\indexSubcarrier]$ as
\if \IEEEsubmission 0
	\begin{align}
		\receivedSignalDiscreteFrequency[\indexSubcarrier] = \channelfreqresponse[\indexSubcarrier]\symbolsInFrequency[\indexSubcarrier]+\noiseDiscrete[\indexSubcarrier]
		=&\sum_{\targetindex=1}^{\numberoftargets}\reflectioncoefficient[\targetindex]|\symbolsInFrequency[\indexSubcarrier]|\constante^{-\constantj2\pi(\fcarrier+\frac{\indexSubcarrier}{\symbolDuration})\timeArrival[\targetindex]+\constantj\angle\symbolsInFrequency[\indexSubcarrier]}+\noiseDiscrete[\indexSubcarrier]\nonumber\\
		=&\sum_{\targetindex=1}^{\numberoftargets}\realPartU[\indexSubcarrier][\targetindex]+\constantj\imagPartV[\indexSubcarrier][\targetindex]+\noiseDiscrete[\indexSubcarrier]~,
		\label{eq:rxFreqSymbolsCRLB}
	\end{align}
\else
	\begin{align}
	\receivedSignalDiscreteFrequency[\indexSubcarrier] = \channelfreqresponse[\indexSubcarrier]\symbolsInFrequency[\indexSubcarrier]+\noiseDiscrete[\indexSubcarrier]
	=\sum_{\targetindex=1}^{\numberoftargets}\reflectioncoefficient[\targetindex]|\symbolsInFrequency[\indexSubcarrier]|\constante^{-\constantj2\pi(\fcarrier+\frac{\indexSubcarrier}{\symbolDuration})\timeArrival[\targetindex]+\constantj\angle\symbolsInFrequency[\indexSubcarrier]}+\noiseDiscrete[\indexSubcarrier]
	=\sum_{\targetindex=1}^{\numberoftargets}\realPartU[\indexSubcarrier][\targetindex]+\constantj\imagPartV[\indexSubcarrier][\targetindex]+\noiseDiscrete[\indexSubcarrier]~,
	\label{eq:rxFreqSymbolsCRLB}
	\end{align}
\fi 
where $\realPartU[\indexSubcarrier][\targetindex]=\reflectioncoefficient[\targetindex]|\symbolsInFrequency[\indexSubcarrier]|\cos(-2\pi(\fcarrier+\frac{\indexSubcarrier}{\symbolDuration})\timeArrival[\targetindex]+\angle\symbolsInFrequency[\indexSubcarrier])$ and $\imagPartV[\indexSubcarrier][\targetindex]=\reflectioncoefficient[\targetindex]|\symbolsInFrequency[\indexSubcarrier]|\sin(-2\pi(\fcarrier+\frac{\indexSubcarrier}{\symbolDuration})\timeArrival[\targetindex]+\angle\symbolsInFrequency[\indexSubcarrier])$. Let $\unknownVector$ be the vector that contains the unknown parameters as $\unknownVector=[\unknownVectorEle[1],\mydots, \unknownVectorEle[2\numberoftargets]]=[\timeArrival[1], \mydots, \timeArrival[\numberoftargets],\reflectioncoefficient[1], \mydots, \reflectioncoefficient[\numberoftargets]]$. The element on the $i$th  row and $j$th column of the $2\numberoftargets\times2\numberoftargets$ \ac{FIM} can then be calculated as
$
\fisherInformationEle[i][j]=\frac{2}{\noiseVariance}\sum_{\targetindex=1}^{\numberoftargets}\sum_{\indexSubcarrier=\lowerFrequency}^{\upperFrequency}\frac{\partial \realPartU[\indexSubcarrier][\targetindex]}{\partial \unknownVectorEle[i]}\frac{\partial \realPartU[\indexSubcarrier][\targetindex]}{\partial \unknownVectorEle[j]}+\frac{\partial \imagPartV[\indexSubcarrier][\targetindex]}{\partial \unknownVectorEle[i]}\frac{\partial \imagPartV[\indexSubcarrier][\targetindex]}{\partial \unknownVectorEle[j]}
$. By evaluating the derivatives, $\fisherInformationEle[i][j]$ can be obtained as
\begin{align}
	\fisherInformationEle[i][j] = 
	\begin{cases}
\frac{8\pi^2\reflectioncoefficient[i]^2}{\noiseVariance}\sum_{\indexSubcarrier=\lowerFrequency}^{\upperFrequency}|\symbolsInFrequency[\indexSubcarrier]|^2(\frac{\indexSubcarrier}{\symbolDuration}+\fcarrier)^2, & i=j\le\numberoftargets\\
		\frac{2\reflectioncoefficient[i-\numberoftargets]^2}{\noiseVariance}\sum_{\indexSubcarrier=\lowerFrequency}^{\upperFrequency}|\symbolsInFrequency[\indexSubcarrier]|^2, & i=j>\numberoftargets\\
		0 & \text{otherwise}
	\end{cases}.
\label{eq:fimele}
\end{align}
The \acp{CRLB} for the unknown parameters are the diagonal elements of the inverse of the \ac{FIM}. Since the \ac{FIM} is a diagonal matrix, the unbiased \ac{CRLB} of the ranges and the \ac{CRLB} of the reflection coefficients are given by
\if \IEEEsubmission 0 
	\begin{align}
	\RMSE^2 \triangleq & \expectedValue[{\sum_{\targetindex=1}^{\numberoftargets}|\distance[\targetindex]-\distanceEst[\targetindex]|^2}]\ge \frac{\speedoflight^2}{4}
	\sum_{i=1}^{\numberoftargets}\frac{1}{\fisherInformationEle[i][i]}\nonumber\\
	=& \frac{\noiseVariance\speedoflight^2}{32\pi^2\sum_{\indexSubcarrier=\lowerFrequency}^{\upperFrequency}|\symbolsInFrequency[\indexSubcarrier]|^2(\frac{\indexSubcarrier}{\symbolDuration}+\fcarrier)^2}\sum_{\targetindex=1}^{\numberoftargets}\frac{1}{\reflectioncoefficient[\targetindex]^2}~,
	\label{eq:rmsedistance}
	\end{align}
\else
	\begin{align}
		\RMSE^2 \triangleq & \expectedValue[{\sum_{\targetindex=1}^{\numberoftargets}|\distance[\targetindex]-\distanceEst[\targetindex]|^2}]\ge \frac{\speedoflight^2}{4}
		\sum_{i=1}^{\numberoftargets}\frac{1}{\fisherInformationEle[i][i]}
		= \frac{\noiseVariance\speedoflight^2}{32\pi^2\sum_{\indexSubcarrier=\lowerFrequency}^{\upperFrequency}|\symbolsInFrequency[\indexSubcarrier]|^2(\frac{\indexSubcarrier}{\symbolDuration}+\fcarrier)^2}\sum_{\targetindex=1}^{\numberoftargets}\frac{1}{\reflectioncoefficient[\targetindex]^2}~,
		\label{eq:rmsedistance}
	\end{align}
\fi
and
\if \IEEEsubmission 0
\begin{align}
	\RMSEref^2 \triangleq & \expectedValue[{\sum_{\targetindex=1}^{\numberoftargets}|\reflectioncoefficient[\targetindex]-\reflectioncoefficientEst[\targetindex]|^2}]\nonumber\ge
	\sum_{i=\numberoftargets+1}^{2\numberoftargets}\frac{1}{\fisherInformationEle[i][i]}\\
	=& \frac{\noiseVariance}{2\sum_{\indexSubcarrier=\lowerFrequency}^{\upperFrequency}|\symbolsInFrequency[\indexSubcarrier]|^2}\sum_{\targetindex=1}^{\numberoftargets}\frac{1}{\reflectioncoefficient[\targetindex]^2}~,
	\label{eq:rmseref}
\end{align}
\else
\begin{align}
	\RMSEref^2 \triangleq  \expectedValue[{\sum_{\targetindex=1}^{\numberoftargets}|\reflectioncoefficient[\targetindex]-\reflectioncoefficientEst[\targetindex]|^2}]\ge
	\sum_{i=\numberoftargets+1}^{2\numberoftargets}\frac{1}{\fisherInformationEle[i][i]}= \frac{\noiseVariance}{2\sum_{\indexSubcarrier=\lowerFrequency}^{\upperFrequency}|\symbolsInFrequency[\indexSubcarrier]|^2}\sum_{\targetindex=1}^{\numberoftargets}\frac{1}{\reflectioncoefficient[\targetindex]^2}~,
	\label{eq:rmseref}
\end{align}
\fi
respectively. By using the fact that $\expectedValue[{|\symbolsInFrequency[\indexSubcarrier]|^2}]=|\FDSScoef[\indexSubcarrier]|^2$, $\fisherInformationEle[i][j]$ can be re-expressed by  replacing $|\symbolsInFrequency[\indexSubcarrier]|^2$ with  $|\FDSScoef[\indexSubcarrier]|^2$ in \eqref{eq:fimele}. Therefore, \eqref{eq:rmsedistance} and \eqref{eq:rmseref} can be modified by replacing $|\symbolsInFrequency[\indexSubcarrier]|^2$ with  $|\FDSScoef[\indexSubcarrier]|^2$.

In the literature, various \acp{CRLB} are derived for different scenarios. For instance, by using \ac{OFDM} with unimodular sequences, the \ac{CRLB} of $\RMSE^2$ was calculated in \cite[Section 3.3.3]{braun_dissertation} as
\begin{align}
	\RMSE^2 \ge & \frac{3\noiseVariance\speedoflight^2}{8\pi^2\numberOfShifts(\numberOfShifts^2-1)}\sum_{\targetindex=1}^{\numberoftargets}\frac{1}{\reflectioncoefficient[\targetindex]^2}.
	\label{eq:rmsedistanceWithoutPhase}
\end{align}
In \cite{Turlapaty_2014},  it was derived when the \ac{OFDM} subcarriers are weighted. 
The main difference between \eqref{eq:rmsedistance} and the \acp{CRLB} derived in these studies is the distance-dependent phase in the channel model. While these studies assume that the phase is unknown and independent from the target's location,  we consider the fact that the phase is a function of the target's distance in  \eqref{eq:radarChannel} \cite{tse_viswanath_2005}.

\section{Numerical Results}
\label{sec:numresults}
We consider IEEE 802.11ay \ac{OFDM} with $4$ channels, where $\symbolDuration\approx194$~ns and $\CPDuration\approx48.48$~ns, $\fcarrier=64.8$~GHz, $\fsample=10.56$~Gsps, $\idftSize=2048$, and $\CPSize=512$. We assume that $\upperFrequency=724$, $\lowerFrequency=-723$, and $\numberOfOccupiedSubcarriers=1382$, and $\numberOfShifts=1536$ for $4$ channels\footnote{The reason for $\numberOfShifts=1536$ is that we can compare the proposed scheme with the \ac{CS-RM} under this configuration.}.  Therefore, the bandwidth of the signal is approximately $7.2$~GHz for \acs{CSC-IM}. The maximum range of the radar is $7.27$~m. The modulation symbols are based on \ac{QPSK}, i.e., $\numberOfPointsForPSK=4$. The \ac{FDSS} coefficients are chosen based on \eqref{eq:fresnekfcn} and \eqref{eq:besselfcn} and assumed to be known at the \ac{RXc}.We consider $\numberofIndices\in\{1,2,5\}$ and set $\separationValue$ to $84$ (i.e., no bits loss based on \figurename~\ref{fig:IS}\subref{subfig:smaxvsm})  for $\numberofIndices=2$  and $252$ (i.e., 10~bits are sacrificed based on \figurename~\ref{fig:IS}\subref{subfig:numberOfBits}) for  $\numberofIndices=5$, when the \ac{IS} is considered. Otherwise, $\separationValue$ is set to $0$. We compare the proposed scheme with three different alternatives: \ac{OFDM-IM}, \ac{DFT-s-OFDM-IM} (i.e., no \ac{FDSS} is applied),  and the \ac{CS-RM} \cite{davis_1999,sahin_2020gm}. For \ac{OFDM-IM}, an \ac{ML} detector that incorporates the channel frequency response is utilized \cite{basar_2013}.   For \ac{CS-RM}, we use a seed \ac{GCP} of length $N=3$ and use $\positiveInt=9$. To  facilitate the ML-based decoder proposed in \cite{sahin_2020gm} for \acp{CS}, we consider only $(\positiveInt-1)!$ cosets. Since these schemes do not use \ac{FDSS}, their bandwidths are approximately $7.9$ GHz. For fading channel,  a power delay profile with three paths where the relative powers are 0~dB, -10~dB, -20~dB at 0~ns, 10~ns, and 20~ns with Rician factors of 10, 0, and 0, respectively, is considered. The number of information bits transmitted are $12$, $24$, and $56$ bits for \ac{IM}-based schemes without \ac{IS} for $\numberofIndices=1,2$, and $5$, respectively. When \ac{IS} is considered,  $24$ bits (i.e., no \ac{SE} loss due to the \ac{IS}) and $46$ bits (i.e., \ac{SE} loss due to the \ac{IS}) are transmitted for $\numberofIndices=2$ and $5$, respectively. With \ac{CS-RM}, $35$ bits are transmitted for each \ac{OFDM} symbol.

\begin{figure}
	\centering
	\includegraphics[width =3.3in]{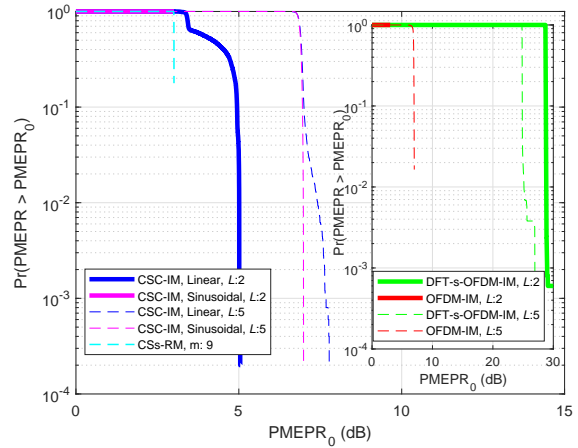}
	\caption{PMEPR distributions for different schemes.}
	\label{fig:pmepr}
\if\IEEEsubmission1
\vspace{-3mm}
\fi
\end{figure}
\subsection{Communications Performance}
In \figurename~\ref{fig:pmepr}, we compare \ac{PMEPR} distributions. The signals are over-sampled to measure \ac{PMEPR} accurately. The \ac{PMEPR} is always less than or equal to $3$ dB for \ac{CS-RM} and sinusoidal chirps for $\numberofIndices=2$.   However, the distortion on linear chirps due to the truncation is higher than the one for sinusoidal chirps. Therefore, the \acp{CS} are not accurately formed with linear chirps under our simulation settings and the maximum \ac{PMEPR} reaches to $5$~dB. 
For $\numberofIndices=5$, the \ac{PMEPR} is still limited for the proposed scheme and  the maximum  \acp{PMEPR}  are $10\log_{10}{5}=6.98$~dB and $7.5$~dB for sinusoidal and linear chirps, respectively. On the other hand, they result in a higher \ac{SE} as compared to \ac{CS-RM}. \ac{OFDM-IM} results in \ac{PMEPR} distributions similar to the ones for the proposed scheme for $\numberofIndices=2$ and $\numberofIndices=5$. However, it does not spread the energy in the frequency domain, which is needed for radar functionality. While \ac{DFT-s-OFDM-IM} spreads the energy in time, it causes signals with very high \acp{PMEPR}. The main reason for this behavior is that \ac{DFT-s-OFDM-IM} actives only $\numberofIndices$ indices that are represented as Dirichlet-sinc pulses in time (see Fig. 5 in \cite{Safi_2020_CCNC}). Therefore, \ac{CSC-IM} is superior to \ac{DFT-s-OFDM-IM}  and \ac{OFDM-IM} for radar applications by reducing \ac{PMEPR} and spreading the energy in both time and frequency.

\begin{figure*}
	\centering
	\subfloat[BLER versus $\EbNO$  in \ac{AWGN}.]{\includegraphics[width =3.3in]{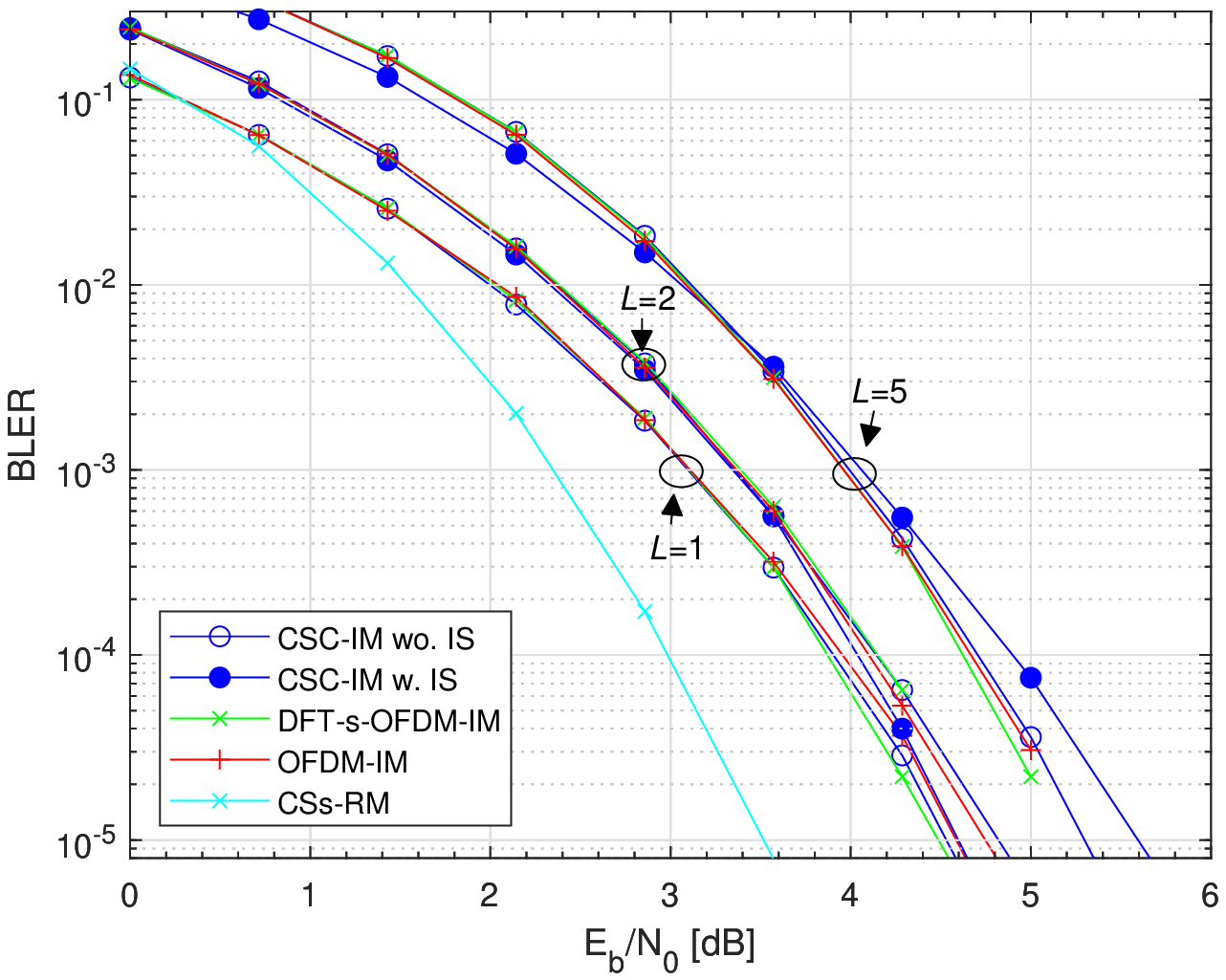}
		\label{subfig:blerEbN0_awgn}}
	\subfloat[BLER versus \ac{SNR} in \ac{AWGN}.]{\includegraphics[width =3.33in]{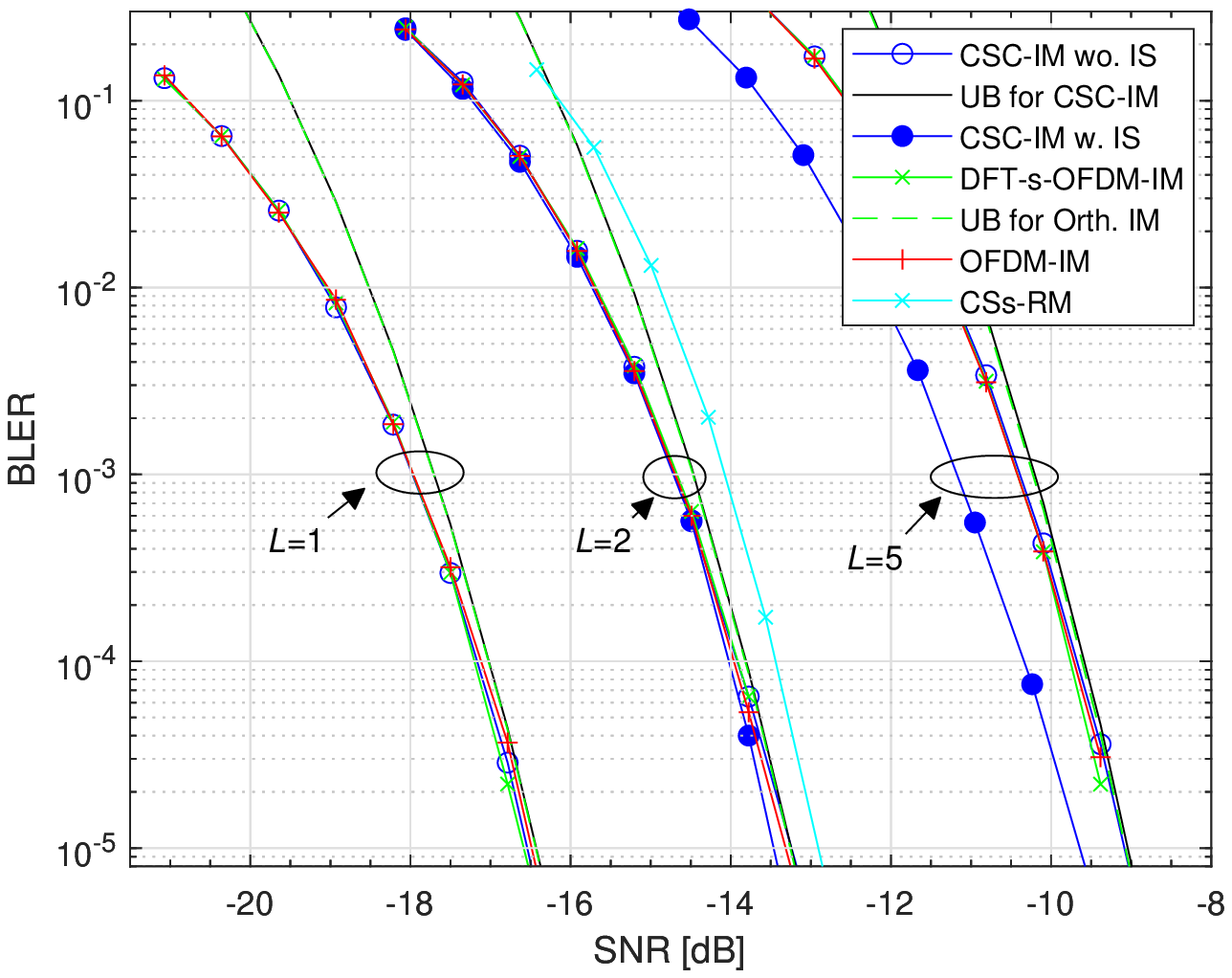}
		\label{subfig:blerEsN0_awgn}}\if\IEEEsubmission1
	\vspace{-3mm}
\fi\\			
	\subfloat[BLER versus $\EbNO$  in fading channel.]{\includegraphics[width =3.3in]{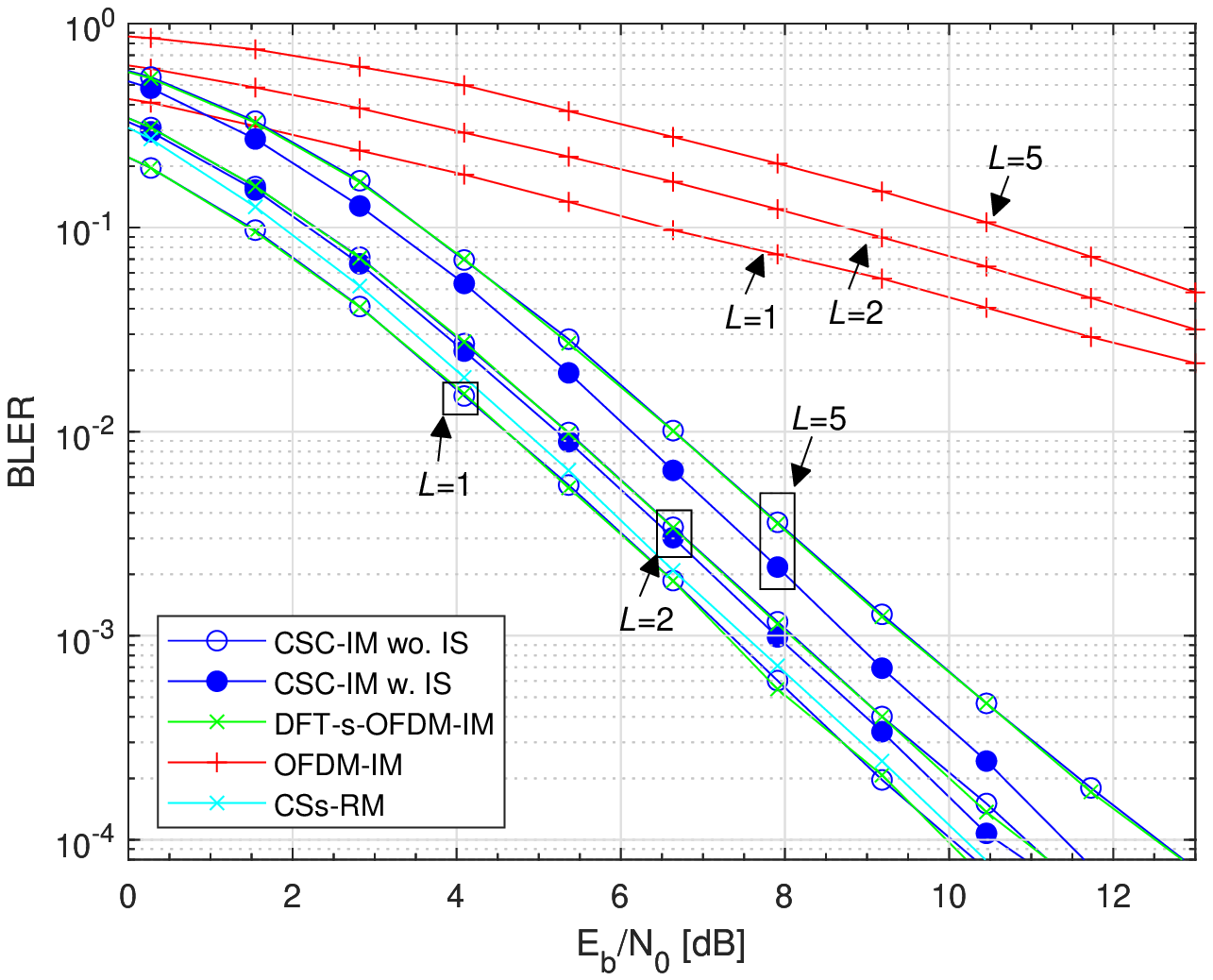}
		\label{subfig:blerEbN0_fading}}
	\subfloat[BLER versus \ac{SNR} in fading channel.]{\includegraphics[width =3.3in]{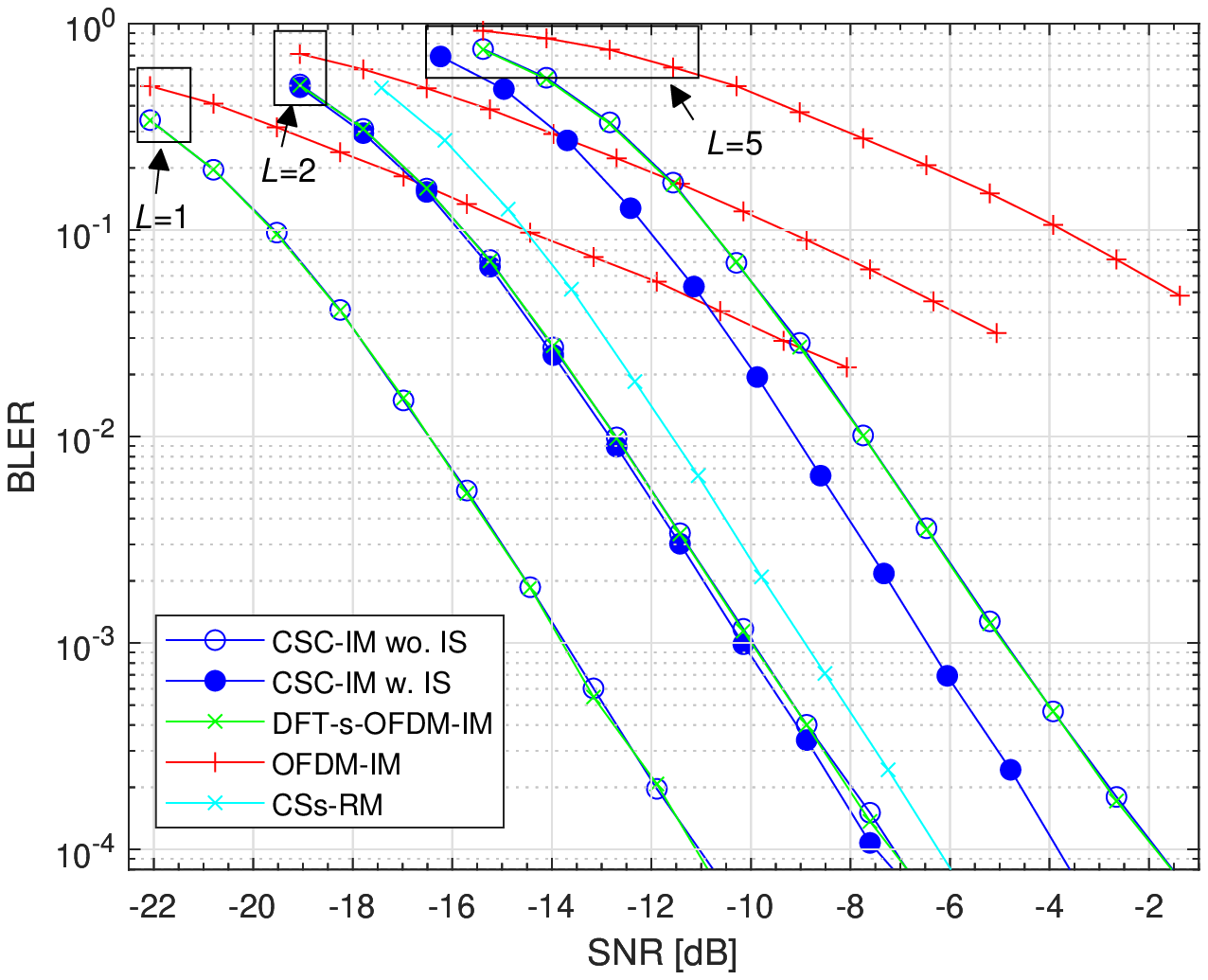}
		\label{subfig:blerEsN0_fading}}~	
	\caption{BLER performance at \ac{RXr} for different schemes.}
	\label{fig:bler}
\if\IEEEsubmission1
\vspace{-3mm}
\fi
\end{figure*}
In \figurename~\ref{fig:bler}, we compare the \ac{BLER} performance in \ac{AWGN} and fading channel. For the proposed scheme, the receiver equalizes the signal even in the \ac{AWGN} channel because of \ac{FDSS}. Since a flatter response improves the \ac{BLER} result in both \ac{AWGN} and fading channels \cite{sahin_2020}, we consider only linear \acp{CSC} for the rest of the analysis. In \figurename~\ref{fig:bler}\subref{subfig:blerEbN0_awgn}, we  provide \ac{BLER} comparisons for a given $\EbNO$. The \ac{CS-RM} is superior to all schemes and provides $1$-$2$~dB gain  at 1e-3 \ac{BLER}. The \ac{CSC-IM} operates in the range of $3$-$4$~dB $\EbNO$  at 1e-3 \ac{BLER} and the performance degrades slightly for a larger $\numberofIndices$. Their error rates are similar to those of \ac{OFDM-IM} and \ac{DFT-s-OFDM-IM}, which shows  the equalization due to the \ac{FDSS} do not degrade the error rate under our simulation settings. In \figurename~\ref{fig:bler}\subref{subfig:blerEsN0_awgn} and \figurename~\ref{fig:bler}\subref{subfig:blerEsN0_fading}, we analyze \ac{BLER} for a given \ac{SNR}. The \acp{BLER} for \ac{CSC-IM} and orthogonal \ac{IM} schemes (i.e., \ac{OFDM-IM} and \ac{DFT-s-OFDM-IM}) approach to the corresponding \acp{UB} given in \eqref{eq:unionbound}, rapidly. As opposed to \ac{CS-RM}, the proposed scheme provides a range of solutions with the various data rates, maximum \acp{PMEPR}, and operating \ac{SNR} points. For example, for $\numberofIndices=1$, it results in $0$~dB \ac{PMEPR} with a very low data rate transmission while it increases the data rate by using $\numberofIndices=5$ chirps at the expense of a higher \ac{PMEPR}. In \figurename~\ref{fig:bler}\subref{subfig:blerEbN0_fading} and \figurename~\ref{fig:bler}\subref{subfig:blerEsN0_fading}, we analyze the same schemes in fading channel.  The performance of \ac{OFDM-IM} is worse than all other schemes since it does not exploit the frequency selectivity. The slopes of the \ac{BLER} curves for \ac{CSC-IM} under  the fading channel are also noticeably higher than the ones for \ac{OFDM-IM}. Although the \ac{DFT-s-OFDM-IM} is similar to \ac{CSC-IM}, a larger power back-off is required for \ac{DFT-s-OFDM-IM} (see \ac{PMEPR} distributions in \figurename~\ref{fig:pmepr}). We also observe that the difference between  \ac{CS-RM} and  \ac{CSC-IM} diminishes in the fading channel and it is less than $1$~dB for $\numberofIndices=2$. It is worth noting that the \ac{ML} detector for \ac{CS-RM} is based on an ML-based algorithm \cite{sahin_2020gm}, which causes a high-complexity detector due to the second-order coset term. On the other hand, the proposed scheme relies on a single $\numberOfShifts$-\ac{IDFT}, per-bin \ac{ML} detection, and recursive mapping rules discussed in Section~\ref{subsubsec:bijection}. We also observe that the \ac{IS} slightly reduces the error rate as in \figurename~\ref{fig:bler}\subref{subfig:blerEsN0_fading} as \ac{IS} limits the search space for indices as discussed in Section~\ref{subsubsec:ISerrorrate}.

\def\distanceBetweenTargets{\Delta r}

\subsection{Radar Performance}
We consider two scenarios for evaluating \ac{RXr} performance. In the first scenario, a single target is assumed. Its location is drawn from a uniform distribution between $2$~m and $3$~m and the true value of the reflection coefficient is set to $-1$, which considers the phase change of a reflected signal \cite{tse_viswanath_2005}. For the second scenario, we consider two targets located nearby. The location of the first target is random between $2$~m and $3$~m and its reflection coefficient, unknown to the \ac{RXr}, is set to $-\sqrt{2}/2$. The second target  with the true value of the reflection coefficient of $-\sqrt{2}/2$ is away from the first target by $\distanceBetweenTargets$, where $\distanceBetweenTargets$ is a randomly chosen between $1.5\minimumResolution\approxeq3.16$~cm and $2\minimumResolution\approxeq4.21$~cm for $\minimumResolution\approxeq2.1$~cm. We consider the proposed scheme with linear chirps and compare it with \ac{CS-RM}. We exclude \ac{OFDM-IM} (as it does not distribute the signal energy in the frequency domain) and \ac{DFT-s-OFDM-IM} (as it causes high \ac{PMEPR}) for radar functionality.

\begin{figure*}
	\centering
	\subfloat[MF-based estimation and a single target.]{\includegraphics[width =3.3in]{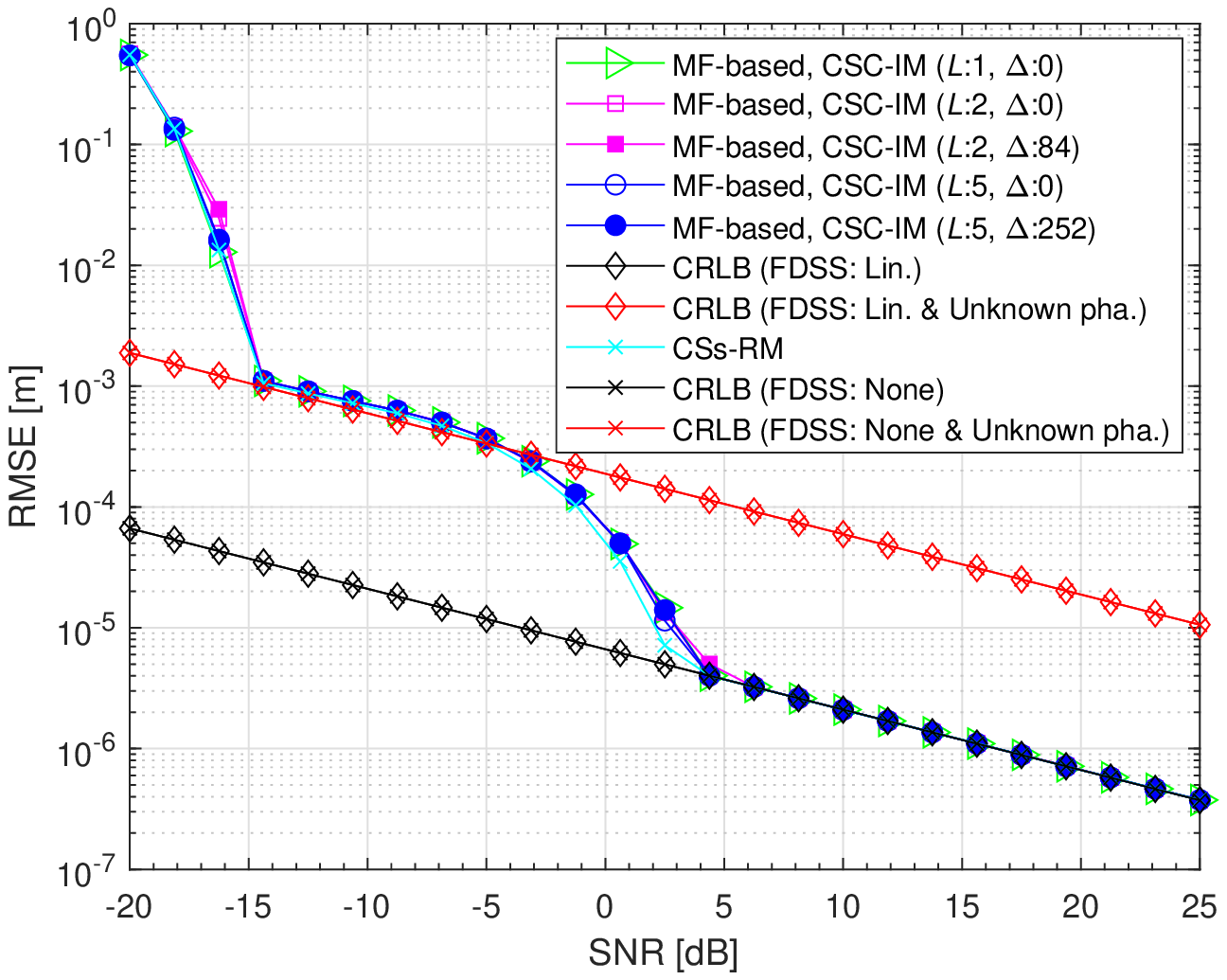}
		\label{subfig:ml_single}}~
	\subfloat[LMMSE-based estimation and a single target.]{\includegraphics[width =3.3in]{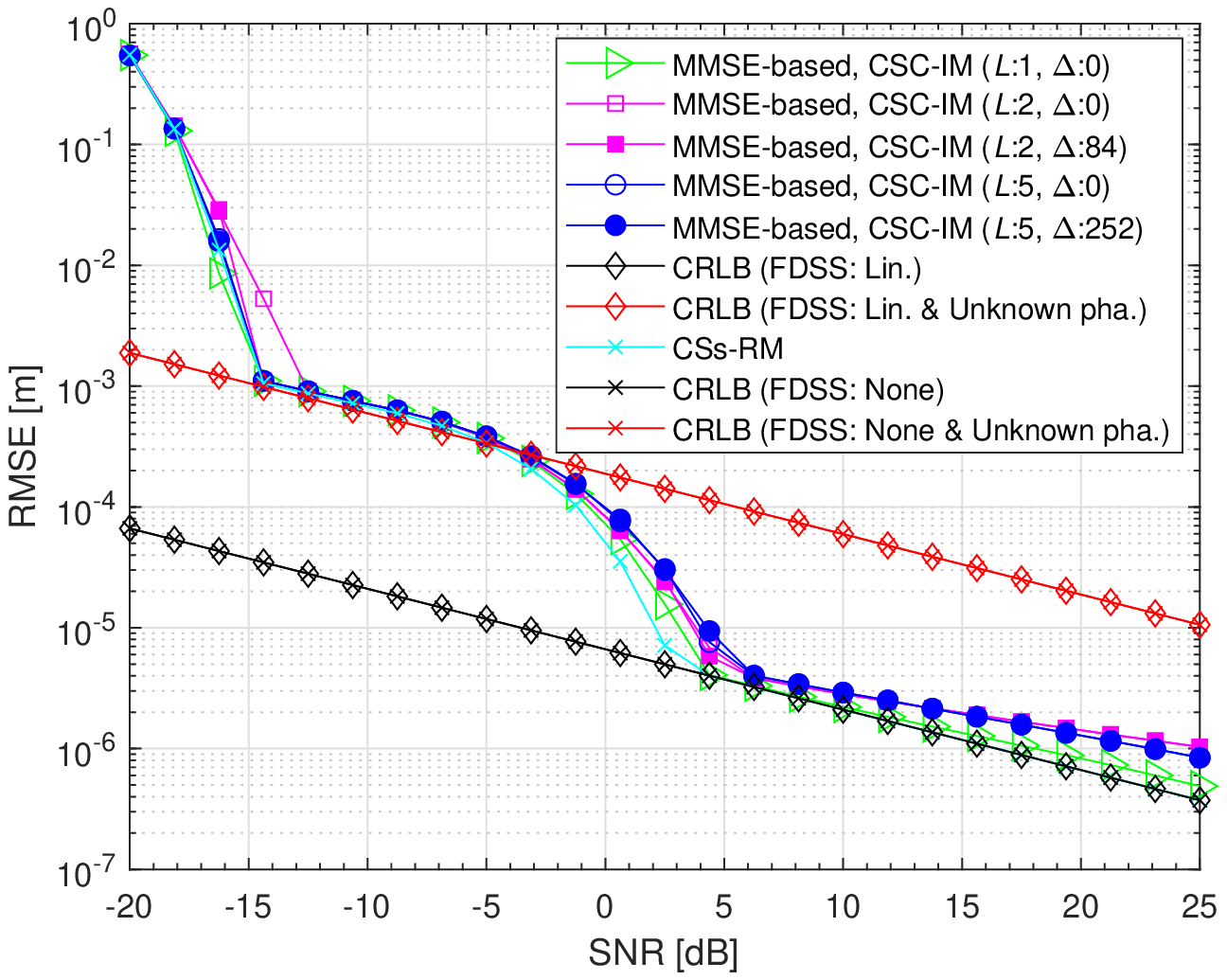}
		\label{subfig:mmse_single}}\if\IEEEsubmission1
	\vspace{-3mm}
\fi
	\\	
	\subfloat[MF-based estimation and two nearby targets.]{\includegraphics[width =3.3in]{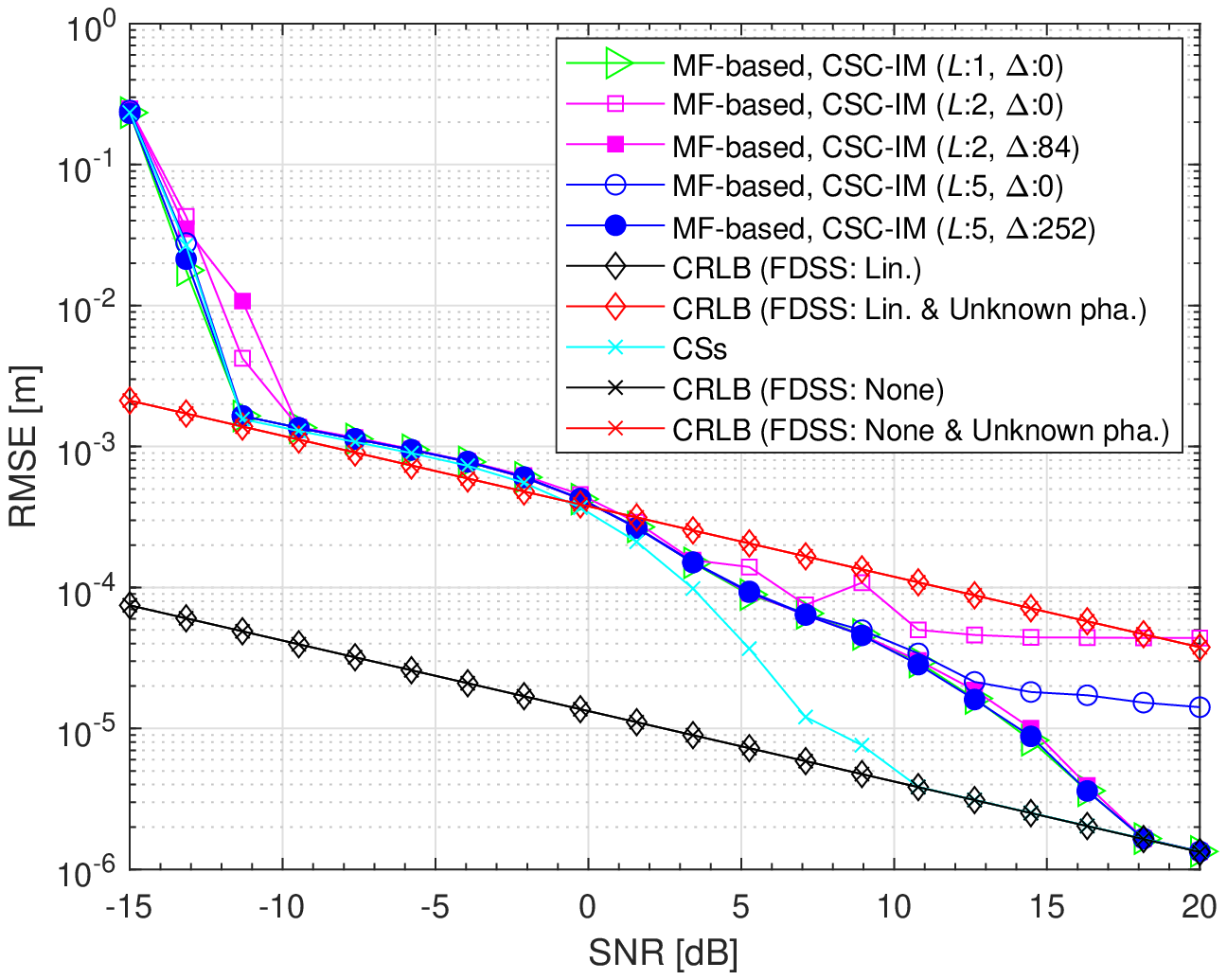}
		\label{subfig:ml_two}}~
	\subfloat[LMMSE-based estimation and two nearby targets.]{\includegraphics[width =3.3in]{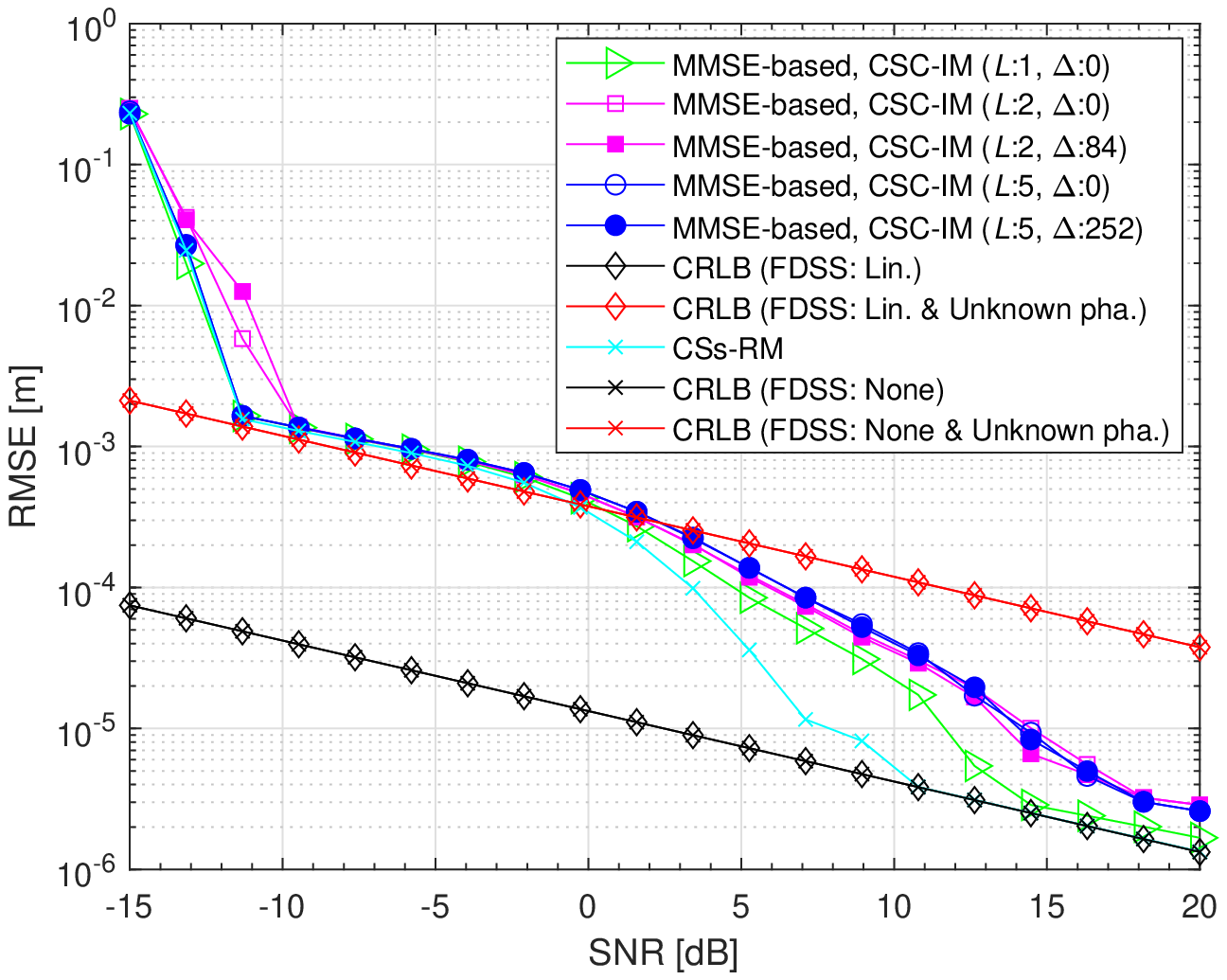}
		\label{subfig:mmse_two}}
	\caption{Accuracy analysis for different estimators and scenarios.}
	\label{fig:rmse}
\if\IEEEsubmission1
\vspace{-3mm}
\fi
\end{figure*}
In \figurename~\ref{fig:rmse}, we evaluate the accuracy of range estimators by providing \ac{RMSE} versus \ac{SNR} curves. In \figurename~\ref{fig:rmse}\subref{subfig:ml_single} and \ref{fig:rmse}\subref{subfig:mmse_single}, we consider the first scenario. The performance of schemes are very similar to each other. The \ac{MF}-based estimator attains the \ac{CRLB} derived in \eqref{eq:rmsedistance} as shown in \figurename~\ref{fig:rmse}\subref{subfig:ml_single}. For the sake of comparison, we also plot the \ac{CRLB} in \eqref{eq:rmsedistanceWithoutPhase} for the case when the phase information is unknown and not a function of the target location. The difference between these two bounds indicates the phase information has a notable impact on the  accuracy, which can be exploited at high \ac{SNR}, i.e., for strong reflections. For the \ac{LMMSE}-based estimator, the bound is only attained for \ac{CS-RM}. This is because \acp{CS} are unimodular in the frequency domain, while the symbol energy is not distributed identically to the frequency bins for linear chirps due to the multiple \acp{CSC} and \ac{FDSS}. In \figurename~\ref{fig:rmse}\subref{subfig:ml_two} and \figurename~\ref{fig:rmse}\subref{subfig:mmse_two}, we consider the second scenario. Since \ac{RXr} estimate targets' locations by using the sequence in the frequency domain, the waveform characteristics in the frequency domain plays a role in the accuracy.  For example, the \ac{CS-RM} are the most prominent ones as it leads to   sequences based on \ac{QPSK}. The subcarriers are populated with arbitrary complex numbers for \ac{CSC-IM}, which degrades the accuracy slightly. When \ac{IS} is not adopted for \ac{ML}-based estimation, \figurename~\ref{fig:rmse}\subref{subfig:ml_two} shows \ac{CSC-IM} saturates and never attains the corresponding \ac{CRLB}. However, when \ac{IS} is utilized, the accuracy of \ac{CSC-IM} with $\numberofIndices={2,5}$ is similar to \ac{CSC-IM} with $\numberofIndices=1$ and attains the \ac{CRLB}. For the \ac{LMMSE}-based estimator, \ac{CRLB} is not attained for \ac{CSC-IM} (due to the \ac{FDSS} with non-unimodular coefficients for \acp{CSC}) while it is achieved with \ac{CS-RM} as in \figurename~\ref{fig:rmse}\subref{subfig:mmse_two}. This result implies that using unimodular sequences is beneficial for radar as it yields superior results with both \ac{LMMSE} and \ac{MF}-based estimators. However, the price paid is a higher-complexity \ac{RXr} for decoding \ac{CS-RM}. Although it is possible to use a lower-complexity code with \ac{QPSK}, it is challenging to address the high \ac{PMEPR} for \ac{OFDM}. From this aspect, \ac{CS-RM} is promising for \ac{OFDM}-based \ac{DFRC} applications. On the other hand, \ac{CSC-IM} is more flexible in terms of size and data rate as compared to \ac{CS-RM}.

In \figurename~\ref{fig:res}, we analyze the resolution for the aforementioned schemes by sweeping the distance between two targets. We fix the \ac{SNR} at $20$ dB. All schemes resolve the targets after the minimum resolution $\minimumResolution=2.1$~cm. On the other hand, \ac{CSC-IM} without \ac{IS} cannot resolve the targets as accurate as \ac{CSC-IM} with \ac{IS} even the distance between targets  is larger than $\minimumResolution$.  The results are in line with the ones as observed in \figurename~\ref{fig:rmse}. For the \ac{LMMSE}-based estimation, they do not attain the \ac{CRLB} except for the \ac{CS-RM} although the accuracy improves after $\minimumResolution$. 
\begin{figure*}
	\centering
	\subfloat[ML-based estimation.]{\includegraphics[width =3.3in]{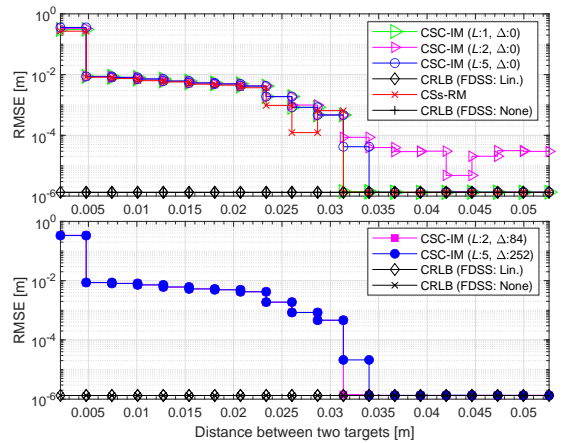}
		\label{subfig:resML}}~		
	\subfloat[LMMSE-based estimation.]{\includegraphics[width =3.3in]{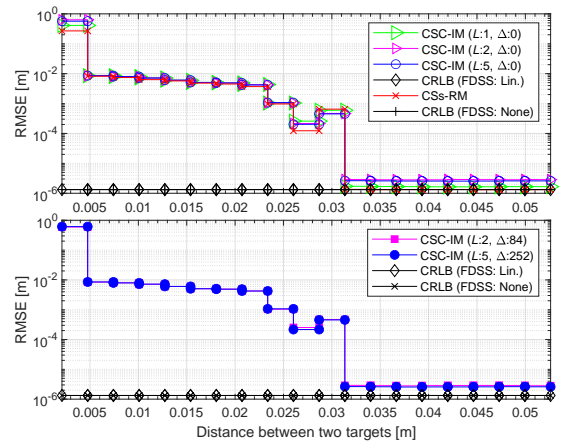}
		\label{subfig:resMMSE}}~
	\caption{Resolution analysis for different estimators.}
	\label{fig:res}
\if\IEEEsubmission1
\vspace{-3mm}
\fi
\end{figure*}

\section{Concluding Remarks}
In this study, we propose \ac{CSC-IM} for \ac{DFRC} scenarios. 
We show that this scheme can generate wideband signals while bringing a trade-off between \ac{SE} and maximum \ac{PMEPR}, i.e., the \ac{SE} increases with the number of \acp{CSC} at the expense of  a higher maximum \ac{PMEPR}. As a special case, we prove that the transmitted signals in the frequency domain lead to new \acp{CS} based on chirps. We also exemplify that Bessel functions and Fresnel integrals can be useful for generating \acp{GCP}. \ac{CSC-IM} is more flexible than the standard \acp{CS} based on \ac{RM} codes in the sense that it allows one to control the maximum \ac{PMEPR} theoretically for the sake of increasing the \ac{SE} while being more flexible in terms of the number of utilized subcarriers. Besides, since \ac{CSC-IM} does not utilize a coset term needed for the \ac{CS-RM}, it enjoys a low-complexity decoder. In this study, we derive the \ac{UB} of the \ac{BLER} for \ac{CSC-IM}, which also captures the analysis for \ac{OFDM-IM} and \ac{DFT-s-OFDM-IM}.
With comprehensive simulations, we show that the \ac{CSC-IM} offers a lower \ac{PMEPR} than \ac{DFT-s-OFDM-IM} while exploiting frequency selectivity as compared to \ac{OFDM-IM}. \ac{CSC-IM} is more suitable for radar functionality as compared to \ac{DFT-s-OFDM-IM} and \ac{OFDM-IM} as it reduces the \ac{PMEPR} while allowing controllable \ac{AC} properties. We consider two range estimation methods: MF-based and LMMSE-based estimations. For the MF-based estimation, we introduce \ac{IS} to generate a low \ac{AC} zone. We investigate the impact of \ac{IS} on \ac{SE} and provide algorithms that map indices to information bits or vice versa. 
 We show that \ac{IS} helps the estimation accuracy to attain the corresponding \ac{CRLB}. 

\bibliographystyle{IEEEtran}
\label{sec:conclusion}

\bibliography{references}

\end{document}